\documentclass[journal]{IEEEtran}

\usepackage{amsfonts}
\usepackage{cite,graphicx,amsmath,amsthm}
\usepackage{subfigure}
\usepackage{fancyhdr}
\usepackage{dsfont}
\usepackage{array,color}
\usepackage{bm}
\usepackage{float}
\usepackage{algorithm}
\usepackage{algpseudocode}
\usepackage[draft=true]{hyperref}

\allowdisplaybreaks[4]

\newtheorem{theorem}{Theorem}

\newtheorem{lemma}{Lemma}

\newtheorem{proposition}{Proposition}

\newtheorem{corollary}{Corollary}

\newtheorem{property}{Property}

\newtheorem{remark}{Remark}

\newtheorem{claim}{Claim}

\begin{document}

\title{{Angle Aware User Cooperation for Secure Massive MIMO in Rician Fading Channel}}

\author{
Shuai~Wang,~\IEEEmembership{Member,~IEEE}, 
Miaowen~Wen,~\IEEEmembership{Senior~Member,~IEEE}, 
Minghua~Xia,~\IEEEmembership{Member,~IEEE}, 
Rui~Wang,~\IEEEmembership{Member,~IEEE}, 
Qi~Hao,~\IEEEmembership{Member,~IEEE}, 
and~Yik-Chung~Wu,~\IEEEmembership{Senior~Member,~IEEE}

\thanks{
The work was supported by the National Natural Science Foundation of China under Grants 61771232 and 61773197, and the Natural Science Foundation of Guangdong Province under Grant 2017A030313335.
The work was also supported by the Natural Science Foundation of Guangdong Province under Grant 2018B030306005.
The work was supported in part by the Major Science and Technology Special Project of Guangdong Province under Grant 2018B010114001, in part by the Fundamental Research Funds for the Central Universities under Grant 191gjc04. (\textit{Corresponding Author: Rui Wang.})

S. Wang and R. Wang are with the Department of Electrical and Electronic Engineering, Southern University of Science and Technology, Shenzhen 518055, China (e-mail: \{wangs3, wang.r\}@sustech.edu.cn).

M. Wen is with the School of Electronic and Information Engineering, South China University of Technology, Guangzhou 510641, China (e-mail: eemwwen@scut.edu.cn).

M. Xia is with the School of Electronics and Information Technology, Sun Yat-sen University, Guangzhou 510006, China, and also with the Southern Marine Science and Engineering Guangdong Laboratory, Zhuhai 519082, China (e-mail: xiamingh@mail.sysu.edu.cn).

Q. Hao is with the Department of Computer Science and Engineering, Southern University of Science and Technology, Shenzhen 518055, China (e-mail: hao.q@sustech.edu.cn).

Y.-C. Wu is with the Department of Electrical and Electronic Engineering, The University of Hong Kong, Hong Kong (e-mail: ycwu@eee.hku.hk).}
}

\maketitle

\begin{abstract}
Massive multiple-input multiple-output communications can achieve high-level security by concentrating radio frequency signals towards the legitimate users.
However, this system is vulnerable in a Rician fading environment if the eavesdropper positions itself such that its channel is highly ``similar'' to the channel of a legitimate user.
To address this problem, this paper proposes an angle aware user cooperation (AAUC) scheme, which avoids direct transmission to the attacked user and relies on other users for cooperative relaying.
The proposed scheme only requires the eavesdropper's angle information, and adopts an angular secrecy model to represent the average secrecy rate of the attacked system.
With this angular model, the AAUC problem turns out to be nonconvex, and a successive convex optimization algorithm, which converges to a Karush-Kuhn-Tucker solution, is proposed.
Furthermore, a closed-form solution and a Bregman first-order method are derived for the cases of large-scale antennas and large-scale users, respectively.
Extension to the intelligent reflecting surfaces based scheme is also discussed.
Simulation results demonstrate the effectiveness of the proposed successive convex optimization based AAUC scheme, and also validate the low-complexity nature of the proposed large-scale optimization algorithms.
\end{abstract}

\begin{IEEEkeywords}
Angle aware user cooperation, first-order method, intelligent reflecting surfaces, massive MIMO, physical-layer security, Rician fading.
\end{IEEEkeywords}

\IEEEpeerreviewmaketitle

\section{Introduction}

\IEEEPARstart{W}hile 5G is expected to be commercially available in 2020, it cannot fully satisfy the performance requirements of massive connectivity \cite{B5G}.
As a result, beyond 5G aims to serve massive devices with lower latency, higher reliability, and better security.
Among the above metrics, security is considered as one of the most challenging problems in wireless communications \cite{security1}, and it is becoming even more imperative in the era of big data \cite{security2}, as the massive data usually contains privacy messages such as personal information and control signals.
Traditionally, security is enforced by adopting key encryption to prevent the eavesdropper from decoding the message transmitted by legitimate users \cite{security3}.
However, it is shown by Wyner \cite{secrecyrate1} that secure communication in the presence of an eavesdropper can be guaranteed without any key encryption,
as long as the users' channels are better than the eavesdropper's channel \cite{secrecyrate1,secrecyrate2,secrecyrate3,secrecyrate4}.
This is called physical-layer security.

\subsection{Motivation}
In order to achieve the above channel condition for physical-layer security, massive multiple-input multiple-output (MIMO) communication is a promising technique, since massive MIMO can concentrate the radio frequency (RF) signals towards the legitimate users while suppressing the RF power leaked to the eavesdropper \cite{massive1,massive2,massive3,massive_security1,massive_security2,massive_security3}.
This leads to a situation that the received signal power at the legitimate user is several orders of magnitude larger than the received signal power at the eavesdropper, thus enabling excellent security without any extra effort \cite{massive_security1,massive_security2,massive_security3}.
However, there are still chances for the eavesdropper to overhear the information via active and passive attacks \cite{massive_security1,massive_security2}.
In an active attack, the eavesdropper adopts pilot contamination to interfere the channel estimation procedure at the base station (BS) \cite{active1}, misleading the BS to transmit signals towards the eavesdropper.
Such an attack has been recently addressed by careful pilot designs \cite{active1,active2,active3}.
Nonetheless, a passive attack is more difficult to deal with, since the passive eavesdropper can hide itself and therefore the BS would have very limited knowledge about the eavesdropper's channel \cite{passive1,passive2,passive3,passive4}.
If the eavesdropper's channel is independent of users' channels, massive MIMO with artificial noise can be used to prevent eavesdropping \cite{passive1,passive2,passive3,passive4}.
But if the system operates in a Rician fading environment (i.e., there exists a line-of-sight (LoS) link \cite{jzhang,jzhang2,hyang}), the eavesdropper can position itself such that its channel to the BS is highly ``similar'' to the channel from a legitimate user to the BS \cite[Sec. IV-A]{massive_security1}, and this passive attack could pose great threats to the massive MIMO system.

\subsection{Contributions}
In this paper, we propose an angle aware user cooperation (AAUC) scheme to combat the above passive attack in massive MIMO systems.
The proposed AAUC scheme avoids direct transmission to the attacked user from BS and counts on other users for cooperative relaying.
Furthermore, in contrast to existing works (e.g., \cite{passive4}) that require full channel information of the eavesdropper, our scheme only requires partial channel information of the eavesdropper.
In particular, the proposed AAUC scheme exploits the information that the eavesdropper wants to have a similar channel to that of a legitimate user, for which its angle of departure would be similar to that of the attacked user.
Based on this angle awareness, the proposed AAUC scheme automatically reduces the transmit powers of users that are close to the potential eavesdropping region, and an optimization problem is formulated to maximize the average secrecy rate subject to the total power constraint.

Nonetheless, the formulated AAUC problem involves numerical integration in the objective function.
To circumvent this obstacle, an angular secrecy model, which matches the numerical integration very well, is proposed.
Based on the secrecy model, the numerical function is converted into an analytical yet nonconvex function, and a successive convex optimization (SCO) algorithm, which is guaranteed to converge to a Karush-Kuhn-Tucker (KKT) solution to the AAUC problem, is derived.
Furthermore, in the large-scale settings, the closed-form solution and the Bregman first-order method (BFOM) are proposed when the number of antennas and the number of users are large, respectively.
The two large-scale methods reduce the computation time by orders of magnitude compared to SCO based AAUC.
Via integration with intelligent reflecting surfaces (IRS), the power cost of AAUC can be further reduced.
Simulation results demonstrate the effectiveness of the proposed SCO based AAUC, and validate the low-complexity nature of the closed-form AAUC and the BFOM based AAUC algorithms.

\subsection{Outline}
The rest of this paper is organized as follows.
System model and problem formulation are described in Section II and Section III, respectively.
The angular secrecy model and the proposed SCO based AAUC scheme are presented in Section IV.
The large-scale optimization algorithms are derived in Section V, and extension to the IRS based scheme is discussed in Section VI.
Finally, numerical results are presented in Section VII, and conclusions are drawn in Section VIII.

\emph{Notation}:
Italic letters, lowercase and uppercase bold letters represent scalars, vectors, and matrices, respectively.
The operators $\textrm{Tr}(\cdot),(\cdot)^{T},(\cdot)^{H}$ and $(\cdot)^{-1}$ take the trace, transpose, Hermitian, and inverse of a matrix, respectively.
The operator $[x]^+=\mathrm{max}(x,0)$.
The symbol $\mathbf{I}_{N}$ represents the $N\times N$ identity matrix.
The symbol $\mathcal{CN}(0,1)$ represents complex Gaussian distribution with zero mean and unit variance, and $\mathcal{U}(x,y)$ represents the uniform distribution within the interval $[x,y]$.
Finally, $\mathbb{E}(\cdot)$ represents the expectation of a random variable, and $\mathrm{exp}(\cdot)$ represents the exponential function of a scalar.

\section{System Model}

\setcounter{secnumdepth}{4}We consider a multicast system with LoS links, which consists of a BS with $N$ antennas, $K$ single-antenna legitimate users, and a single-antenna eavesdropper.
As shown in Fig. 1, the BS intends to multicast common information to the $K$ legitimate users, while the eavesdropper intends to overhear the signal\footnote{The detection of a passive attack is based on spectrum sensing \cite{detection}.}.
In particular, the eavesdropper hides at the line segment between the BS and one of the users, denoted as user $K$.
In this way, the eavesdropper can receive a highly correlated signal with that of user $K$ \cite{massive_security1,massive_security2}, thus decoding the transmitted information.
To address this problem, the AAUC scheme consisting of two phases with equal duration is proposed, and the details are given below.

\begin{figure}[!t]
\centering
\includegraphics[width=60mm]{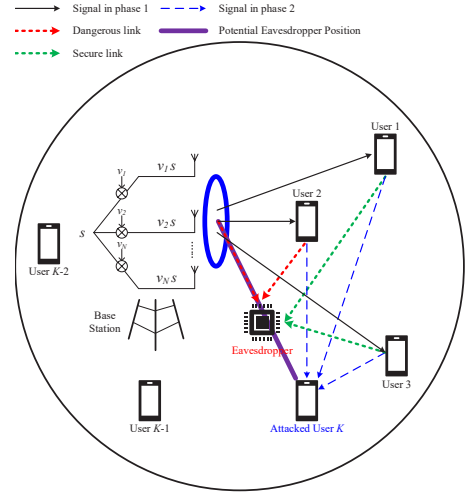}
\caption{System model of massive MIMO communication under the LoS attack.}
\end{figure}

\subsection{Multicasting Phase}

In the multicasting phase, the BS transmits a signal $s$ with $\mathbb{E}[|s|^2]=1$ to users $\{1,2,\cdots,K\}$ through a transmit beamforming vector $\mathbf{v}\in\mathbb{C}^{N\times 1}$ with power $||\mathbf{v}||^2_2$.
Accordingly, the received signal at a legitimate user $k$ is
$\mathbf{g}^H_{k}\mathbf{v}s+n_{k}$,
where $\mathbf{g}_{k}\in \mathbb{C}^{N\times 1}$ is the channel vector from the BS to the user $k$,
and $n_{k}\in \mathbb{C}$ is the zero-mean Gaussian noise at the user $k$ with power $\sigma_k^2$.
Denoting $D_k$ and $\theta_k$ as the distance and the azimuth angle (i.e., $0$ means ``north'' and $\pi/2$ means ``east'') between BS and user $k$, the channel $\mathbf{g}_{k}$ obeys the following Rician fading \cite{channel1,channel2}:
\begin{align}
\mathbf{g}_{k}&=\sqrt{\varrho_0\,\left(\frac{D_{k}}{d_0}\right)^{-\alpha}}\,\Bigg(\sqrt{\frac{K_R}{1+K_R}}\,\mathbf{g}_{k}^{\mathrm{LOS}}
\nonumber\\
&\quad{}
+\sqrt{\frac{1}{1+K_R}}\,\mathbf{g}_{k}^{\mathrm{NLOS}}\Bigg),~~k=1,\cdots,K, \label{racian}
\end{align}
where $\varrho_0$ is the pathloss at distance $d_0=1\,\mathrm{m}$, and $\alpha$ is the pathloss exponent.
Notice that $K_R$ is the Rician K-factor accounting for propagation effects of the LoS and non-LoS links, which can be pre-determined from a few channel measurements in the environment \cite{channel3}.
Furthermore, the LoS component $\mathbf{g}_{k}^{\mathrm{LOS}}$ is \cite{los}
\begin{align}\label{LOS}
\mathbf{g}_{k}^{\mathrm{LOS}}&=\Big[1,\mathrm{exp}\left(-\mathrm{j}\pi\,\mathrm{sin}\,\theta_k\right)
,\cdots,
\nonumber\\
&\quad{}
\mathrm{exp}\left(-(N-1)\,\mathrm{j}\pi\,\mathrm{sin}\,\theta_k\right)\Big]^T,
\end{align}
and the non-LoS component $\mathbf{g}_{k}^{\mathrm{NLOS}}\sim\mathcal{CN}(\mathbf{0},\mathbf{I}_N)$.

Since the angle of eavesdropper $\theta_E$ can be approximated as $\theta_E\approx\theta_K$, $\mathbf{g}_{E}^{\mathrm{LOS}}\approx\mathbf{g}_{K}^{\mathrm{LOS}}$ holds, where $\mathbf{g}_{E}^{\mathrm{LOS}}$ is the LoS component of the channel $\mathbf{g}_{E}\in \mathbb{C}^{N\times 1}$ from the BS to the eavesdropper.
According to $\mathbf{g}_{E}^{\mathrm{LOS}}\approx\mathbf{g}_{K}^{\mathrm{LOS}}$ and \eqref{racian}, we have $\mathbf{g}_{E}//\mathbf{g}_{K}$ due to $K_R\gg1$ in LoS environment, where $\mathbf{a}//\mathbf{b}$ means that vectors $\mathbf{a}$ and $\mathbf{b}$ have the same direction.
To ensure that $|\mathbf{g}^H_{E}\mathbf{v}|=0$, the BS needs to design $\mathbf{v}$ such that $|\mathbf{g}^H_{K}\mathbf{v}|=0$.
Therefore, the achievable rates at the user $K$ and the eavesdropper can be approximated to zero.
On the other hand, the achievable rate $R_k$ at user $k\neq K$ is
\begin{align}\label{Rk}
R_k=\frac{1}{2}\mathrm{log}_2\left(1+\frac{|\mathbf{g}^H_{k}\mathbf{v}|^2}{\sigma_k^2}\right),~~\forall k\neq K,
\end{align}
where the factor $1/2$ is due to the two transmission phases.

\subsection{User Cooperation Phase}

Due to $|\mathbf{g}^H_{K}\mathbf{v}|=0$, the user $K$ would hardly receive information during the multicasting phase.
As a result, the other users need to help to forward the information to user $K$.
In particular, the helping users $\{1,\cdots,K-1\}$ transmit the signal $s$ to the user $K$ through a beamforming vector $\mathbf{w}\in\mathbb{C}^{(K-1)\times 1}=[w_1,...,w_{K-1}]^T$ in the user cooperation phase, and the received signal at the user $K$ is given by
$\mathbf{h}_K^H\mathbf{w}s+n_{K}$,
where $\mathbf{h}_K\in \mathbb{C}^{(K-1)\times1}$ is the channel vector\footnote{Channel $\mathbf{h}_K$ can be estimated at the user $K$, who subsequently forwards $\mathbf{h}_K$ to the BS.} from the helping users to the user $K$.
Therefore, the achievable rate $R_K$ at user $K$ is \cite{relay}
\begin{align}\label{RK}
&R_K=\mathrm{min}\left\{R_1,\cdots,R_{K-1},\frac{1}{2}\mathrm{log}_2\left(1+\frac{|\mathbf{h}_K^H\mathbf{w}|^2}{\sigma_K^2}\right)\right\}.
\end{align}
On the other hand, the received signal at the eavesdropper is given by
$\mathbf{h}_E^H\mathbf{w}s+n_{E}$, where $\mathbf{h}_E\in \mathbb{C}^{(K-1)\times1}$ is the channel vector from the helping users to the eavesdropper and $n_{E}\in \mathbb{C}$ is the zero-mean Gaussian noise at the eavesdropper with power $\sigma_E^2$.
Accordingly, the eavesdropping rate is\footnote{In practice, the eavesdropper may not know $\mathbf{w}$. However, the eavesdropper can adjust its receiver to fully exploit the received signal strength $|\mathbf{h}_E^H\mathbf{w}|^2$.}
\begin{align}\label{RE}
&R_E=\mathrm{min}\left\{R_1,\cdots,R_{K-1},\frac{1}{2}\mathrm{log}_2\left(1+\frac{|\mathbf{h}_E^H\mathbf{w}|^2}{\sigma_E^2}\right)\right\}.
\end{align}

Finally, combining the results in \eqref{Rk}--\eqref{RE}, the secrecy rate of the multicast system under the AAUC scheme can be expressed as \cite{secrecyrate1,secrecyrate2,secrecyrate3}:
\begin{align}\label{secrecyrate}
&\left[\mathop{\mathrm{min}}_{k=1,\cdots,K}R_k-R_E\right]^+
=\Bigg[R
-\frac{1}{2}\mathrm{log}_2\left(1+\frac{|\mathbf{h}_E^H\mathbf{w}|^2}{\sigma_E^2}\right)
 \Bigg]^+,
\end{align}
where the equality is due to $R_K\leq R_k$ for any $k\neq K$ according to \eqref{RK}, and
\begin{align}\label{R_expression}
R
&=\frac{1}{2}\mathrm{min}\Bigg\{\mathrm{log}_2\left(1+\frac{|\mathbf{g}^H_{1}\mathbf{v}|^2}{\sigma_1^2}\right),\cdots,
\nonumber\\
&\quad{}
\mathrm{log}_2\left(1+\frac{|\mathbf{g}^H_{K-1}\mathbf{v}|^2}{\sigma_{K-1}^2}\right),\mathrm{log}_2\left(1+\frac{|\mathbf{h}_K^H\mathbf{w}|^2}{\sigma_K^2}\right)\Bigg\}.
\end{align}

\section{Problem Formulation}

\subsection{Average Secrecy Rate}
The major challenge to maximize \eqref{secrecyrate} is that $\mathbf{h}_E$ is unknown.
A traditional way is to model $\mathbf{h}_E$ as a constant value $\widehat{\mathbf{h}}_E$ plus some random perturbations \cite{robust1,passive4}.
However, such a method requires the knowledge of $\widehat{\mathbf{h}}_E$, which is difficult to obtain since the eavesdropper is passive.
To combat the passive attack without knowing the full channel information, we propose to exploit the angular information.
In particular, denoting the location of BS as $(0,0)$, the location of user $k$ can be expressed as $\left(D_k\,\mathrm{cos}\,\theta_k,D_k\,\mathrm{sin}\,\theta_k\right)$, where $D_k\,\mathrm{cos}\,\theta_k$ is the latitude and $D_k\,\mathrm{sin}\,\theta_k$ is the longitude.
Since the eavesdropper hides at the line segment between BS and user $K$,
the location of eavesdropper can be represented as $\left(\rho\,\mathrm{cos}\,\theta_K,\rho\,\mathrm{sin}\,\theta_K\right)$, where $\rho$ is bounded from $0$ to $D_K$.

Based on the locations, the distance from user $k$ to eavesdropper  is given by
\begin{align}\label{dk}
d_{E,k}(\rho)&=\Big[(\rho\,\mathrm{cos}\,\theta_K-D_k\,\mathrm{cos}\,\theta_k)^2
\nonumber\\
&\quad{}
+(\rho\,\mathrm{sin}\,\theta_K-D_k\,\mathrm{sin}\,\theta_k)^2\Big]^{-1/2}.
\end{align}
By further adopting the Rician fading channel model \cite{channel1,channel2}, the channel $h_{E,k}$ (which is the $k^{\mathrm{th}}$ element in $\mathbf{h}_E=[h_{E,1},\cdots,h_{E,K-1}]^T$) from user $k$ to eavesdropper can be modeled as:
\begin{align}\label{channel}
h_{E,k}(\rho,\phi_k,\delta_k)&=\sqrt{\varrho_0\left(\frac{d_{E,k}(\rho)}{d_0}\right)^{-\alpha}}
\nonumber\\
&
\times
\left(\sqrt{\frac{K_R}{1+K_R}}\,\mathrm{e}^{\mathrm{j}\phi_k}
+\sqrt{\frac{1}{1+K_R}}\, \delta_{k}\right),
\end{align}
where $\phi_k$ is the phase\footnote{Since each user has its own local phase shift in the transmitted signal, $\phi_k$ is a random value independent of $\rho$.} of LoS link with $\phi_k\sim\mathcal{U}(-\pi,+\pi)$, and $\delta_{k}$ is the non-LoS component with $\delta_{k}\sim\mathcal{CN}(0,1)$.

By substituting the channel model \eqref{channel} into \eqref{secrecyrate}, the secrecy rate for a fixed $(\rho,\{\phi_k,\delta_k\})$ can be obtained.
Then by averaging the result in \eqref{secrecyrate} over $(\rho,\{\phi_k,\delta_k\})$, the average secrecy rate of the system can be computed as follows \cite{averagesecrecy1,averagesecrecy2, averagesecrecy3}:
\begin{align}\label{Sk}
&S(R,\mathbf{w}|\sigma_E)=\frac{1}{D_K}\int_{0}^{D_K}\mathbb{E}_{\{\phi_k,\delta_k\}}\Bigg\{\Bigg[R
\nonumber\\
&
-\frac{1}{2}\mathrm{log}_2
\left(1+\frac{\Big|\mathbf{h}^H_E(\rho,\{\phi_k,\delta_k\})\mathbf{w}
\Big|^2}{\sigma^2_E}\right)\Bigg]^+\Bigg\}\,\mathrm{d}\rho.
\end{align}

\subsection{Secrecy Rate Maximization}

In the considered system, the design variables that can be controlled are the transmit beamformer $\mathbf{v}$ at BS and the distributed beamformer $\mathbf{w}$ at helping users.
Since the power costs at both BS and helping users should be smaller than the total budget, the beamformers need to satisfy
$\frac{||\mathbf{v}||^2}{2}+\frac{||\mathbf{w}||^2}{2}\leq P_{\rm{max}}$, where the factor $1/2$ is due to the two transmission phases, and $P_{\rm{max}}$ is the total transmit power budget.
Having the transmit power satisfied, it is then crucial to maximize the average secrecy rate $S(R,\mathbf{w}|\sigma_E)$, which leads to the following optimization problem:
\begin{subequations}
\begin{align}
\mathrm{P}1:\mathop{\mathrm{max}}_{\substack{\mathbf{v},\mathbf{w},R}}
~~&\frac{1}{D_K}\int_{0}^{D_K}\mathbb{E}_{\{\phi_k,\delta_k\}}\Bigg\{\Bigg[R-\frac{1}{2}\mathrm{log}_2
\Big(1
\nonumber\\
&
+\frac{\Big|\mathbf{h}^H_E(\rho,\{\phi_k,\delta_k\})\mathbf{w}
\Big|^2}{\sigma^2_E}\Big)\Bigg]^+\Bigg\}\,\mathrm{d}\rho,
\nonumber\\
~~~~~~\mathrm{s. t.}~~~~&\frac{||\mathbf{v}||^2_2}{2}+\frac{||\mathbf{w}||^2_2}{2}\leq P_{\rm{max}}, \label{P1a}
\\
&|\mathbf{g}^H_{K}\mathbf{v}|=0, \label{P1b}
\\
&R\leq\frac{1}{2}\mathrm{log}_2\left(1+\frac{|\mathbf{g}^H_{k}\mathbf{v}|^2}{\sigma_k^2}\right),~~\forall k\neq K,
\label{P1c}
\\
&R\leq\frac{1}{2}\mathrm{log}_2\left(1+\frac{|\mathbf{h}_K^H\mathbf{w}|^2}{\sigma_K^2}\right),
\label{P1d}
\end{align}
\end{subequations}
where the objective function of $\mathrm{P}1$ is obtained from \eqref{Sk}.
The constraint \eqref{P1a} is the operation budget constraint and the constraint \eqref{P1b} guarantees that no information is leaked to the eavesdropper during the multicasting phase.
Finally, the constraints \eqref{P1c} and \eqref{P1d} are the data-rate constraints obtained from \eqref{R_expression}.

Notice that problem $\rm{P}1$ can be extended to the case with different information signals for different users by adding an additional individual transmission phase for the users $\{1,\cdots,K-1\}$.
In addition, the problem $\rm{P}1$ can be extended to the case with separate transmit power constraints at BS and users.
In such a case, the design variables $\mathbf{v}$ and $\mathbf{w}$ are no longer coupled in the constraints, and the proposed algorithms are still applicable to the resultant problem.

\section{Proposed SCO Based AAUC}
\subsection{Angular Secrecy Model}
To solve $\mathrm{P}1$, the first challenge comes from the integration over $\rho$ and the expectation over $\{\phi_k,\delta_k\}$, which make the objective a numerical function of $\mathbf{w}$.
To address this challenge, it is observed from \eqref{Sk} that the secrecy rate $S$ satisfies the following properties:
\begin{itemize}
\item[(i)] When $\sigma_E^2\rightarrow 0$, we must have $S\rightarrow 0$.

\item[(ii)] When $\sigma_E^2\rightarrow +\infty$, the secrecy rate $S$ would reach its maximum value $R$.

\item[(iii)] The secrecy rate $S$ is lower bounded by a logarithm function $\left[R-\frac{1}{2}\mathrm{log}_2
\left(1+\mathbf{w}^H\mathbf{J}\mathbf{w}/\sigma^2_E\right)\right]^+$ (proved in Appendix A), where
$\mathbf{J}=\mathrm{diag}\left(J_1,\cdots,J_{K-1}\right)$ is the angular matrix with
\begin{align}\label{Jk}
J_k=&
\frac{\varrho_0}{D_k}\int_{0}^{D_K}
\Big[(\rho\,\mathrm{cos}\,\theta_K-D_k\,\mathrm{cos}\,\theta_k)^2
\nonumber\\
&
+(\rho\,\mathrm{sin}\,\theta_K-
D_k\,\mathrm{sin}\,\theta_k)^2\Big]^{-\alpha/2}\mathrm{d}\rho.
\end{align}
\end{itemize}
Based on (i)--(iii), the average secrecy rate should have a logarithm curve with respect to noise power $\sigma_E^2$.
To this end, we propose an angular secrecy model as follows:
\begin{align}\label{smodel}
&\widehat{S}\left(R,\mathbf{w}|\sigma_E,\lambda\right)=\left[R-\frac{1}{2}\mathrm{log}_2
\left(1+\frac{\lambda\mathbf{w}^H\mathbf{J}\mathbf{w}}{\sigma^2_E}\right)\right]^+,
\end{align}
where $\lambda\in[0,1]$ is a tuning parameter.
It can be verified that the nonlinear function in \eqref{smodel} satisfies all the properties (i)--(iii).
Moreover, since equation \eqref{Jk} does not involve $\mathbf{w}$ inside the integral, the angular matrix $\mathbf{J}$ can be computed via one-dimensional integration.

\begin{figure}
\centering
\subfigure[]{\includegraphics[width=60mm]{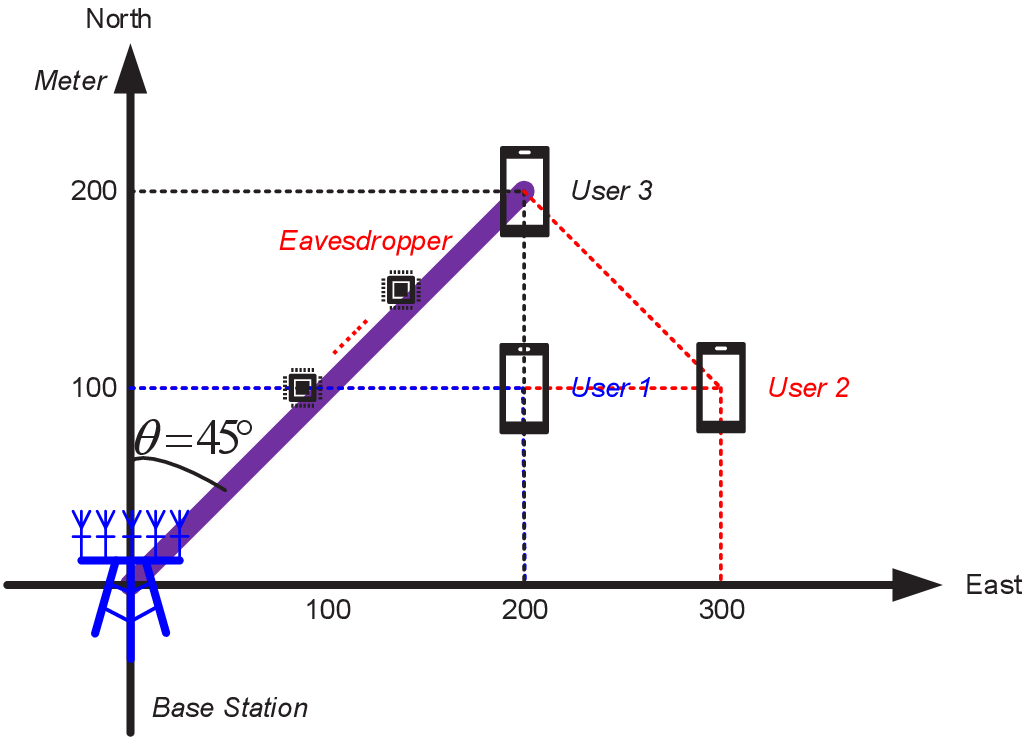}}
\subfigure[]{\includegraphics[width=60mm]{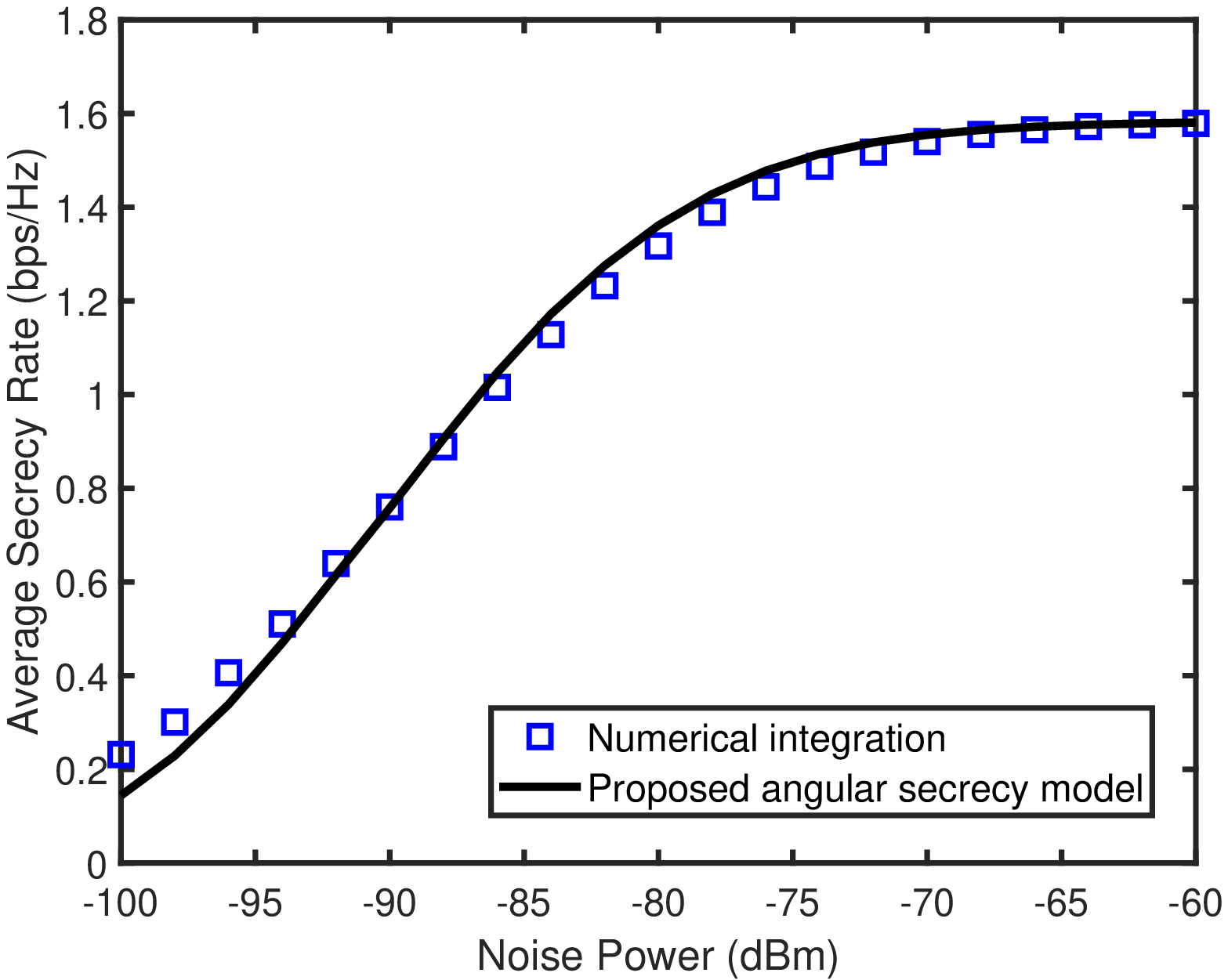}}
\caption{a) Illustrating example with $K=3$, $\alpha=2.5$ and $K_R=30~\mathrm{dB}$; b) Fitting performance of the angular secrecy model.}
\end{figure}

\textbf{Interpretation of the Angular Matrix.} $J_k$ in \eqref{Jk} is the average pathloss from user $k$ to the eavesdropper and quantitatively measures how dangerous for user $k$ to transmit a signal.
For example, in Fig. 2a with $K=3$ and $\alpha=2.5$, we have $J_1=10^{-5}\varrho_0$ and $J_2=2\times 10^{-6}\varrho_0$ according to \eqref{Jk}, meaning that user $2$ can provide much better security.
This corroborates the fact that user $1$ is closer to the line segment between BS and user $3$.

Based on (i)--(iii), $\widehat{S}_{\lambda}\left(R,\mathbf{w}|\sigma_E,\lambda\right)$ would have a similar trend to $S$, and the next step is to minimize the gap between the two curves $\widehat{S}\left(R,\mathbf{w}|\sigma_E,\lambda\right)$ and $S\left(R,\mathbf{w}|\sigma_E\right)$.
Since $\widehat{S}$ is a decreasing function of $\lambda$, one-dimensional search over $[0,1]$ can be adopted to determine an appropriate $\lambda$.
More specifically, we numerically compute $S(R^{(i)},\mathbf{w}^{(i)}|\sigma_E^{(i)})$ using \eqref{Sk} for the input $(R^{(i)},\mathbf{w}^{(i)},\sigma_E^{(i)})$, where $i=1,\cdots,T$ is the index of training sample and $T$ is the sample size.
With the training set $\{(R^{(i)},\mathbf{w}^{(i)},\sigma_E^{(i)})\}_{i=1}^T$, the parameter $\lambda$ in $\widehat{S}$ can be found via nonlinear least squares fitting:
\begin{align}
&\mathop{\mathrm{min}}_{\lambda\in[0,1]}~\frac{1}{T}\mathop{\sum}_{i=1}^T\Big|S(R^{(i)},\mathbf{w}^{(i)}|\sigma_E^{(i)})-\widehat{S}\left(R^{(i)},\mathbf{w}^{(i)}|\sigma_E^{(i)},\lambda\right)\Big|^2. \label{train}
\end{align}

In order to verify the accuracy of the model, we consider the setting of Fig. 2a.
In this environment, we consider $41$ different values of $[\sigma_E^{(i)}]^2$ ranging from $-100\,\mathrm{dBm}$ to $-60\,\mathrm{dBm}$ by a step size of $2\,\mathrm{dBm}$, and $T=21\times 100$ training samples are generated with $R^{(i)}\sim\mathcal{U}(0,3)$ and $\mathbf{w}^{(i)}\sim\mathcal{CN}(\bm{0},\frac{0.01}{K-1}\,\mathbf{I}_{K-1})$ (i.e., the average power of $\mathbf{w}^{(i)}$ is $10\,\mathrm{dBm}$).
It can be seen from Fig. 2b that with the choice of $\lambda=0.64$, the proposed model matches the numerical simulation of $S$ very well under a large range of noise power $\sigma_E^2$.
Moreover, the mean square error of \eqref{train} at $\lambda=0.64$ is $0.0056$, which is negligible compared to the absolute secrecy rate.

\emph{Remark 1:} Notice that the proposed method can be well extended to the 3D scenarios when the eavesdropper is not on the same plane as the users (e.g., the eavesdropper is a flying unmanned aerial vehicle \cite{3d}).
In such a case, the distance from eavesdropper to the user $k$ is
\begin{align}
d_{E,k}(\rho)&=\Big[(\rho\,\mathrm{cos}\,\theta_K-D_k\,\mathrm{cos}\,\theta_k)^2
\nonumber\\
&
+(\rho\,\mathrm{sin}\,\theta_K-D_k\,\mathrm{sin}\,\theta_k)^2+\rho^2\,\mathrm{tan}^2\,\beta\Big]^{-1/2},
\end{align}
where $\beta$ is the elevation angle of the eavesdropper.
Accordingly, the $k^{\mathrm{th}}$ element of the angular matrix becomes
\begin{align}
J_k=&
\frac{\varrho_0}{D_K}\int_{0}^{D_K}
\Big[(\rho\,\mathrm{cos}\,\theta_K-D_k\,\mathrm{cos}\,\theta_k)^2
\nonumber\\
&
+(\rho\,\mathrm{sin}\,\theta_K-
D_k\,\mathrm{sin}\,\theta_k)^2+\rho^2\,\mathrm{tan}^2\,\beta\Big]^{-\alpha/2}\mathrm{d}\rho.
\end{align}
Once the angular matrix is defined, all the derivations in the manuscript stay the same for 2D and 3D scenarios.

\subsection{Proposed SCO Algorithm}

With the angular secrecy model, we can replace $S(R,\mathbf{w}|\sigma_E)$ with $\widehat{S}(R,\mathbf{w}|\sigma_E,\lambda)$ in the objective function of $\mathrm{P}1$, and the resultant problem is given by
\begin{align}
\mathop{\mathrm{max}}_{\substack{\mathbf{v},\mathbf{w},R}}
~&\widehat{S}\left(R,\mathbf{w}|\sigma_E,\lambda\right)
\quad \mathrm{s. t.}\quad\eqref{P1a}-\eqref{P1d}. \label{Plb}
\end{align}
To solve \eqref{Plb}, the following transformations are used to eliminate the constraints.
\begin{itemize}
  \item QR decomposition is applied to $\mathbf{g}_K$, and the unitary matrix $\mathbf{U}\in\mathbb{C}^{N \times (N-1)}$, which spans the null space of $\mathbf{g}_K$, is obtained. With $\mathbf{U}$, we can set $\mathbf{v}=\mathbf{U}\mathbf{z}$, where $\mathbf{z}\in\mathbb{C}^{(N-1) \times 1}$ is a newly introduced variable.
  \item Putting $\mathbf{v}=\mathbf{U}\mathbf{z}$ into \eqref{P1a}, this constraint becomes $||\mathbf{z}||^2_2+||\mathbf{w}||^2_2\leq 2P_{\rm{max}}$, where we have used $\mathbf{U}^H\mathbf{U}=\mathbf{I}$.
  \item Putting $\mathbf{v}=\mathbf{U}\mathbf{z}$ into \eqref{P1b}, we have $|\mathbf{g}^H_{K}\mathbf{U}\mathbf{z}|=0$. This constraint is always feasible due to $\mathbf{g}^H_{K}\mathbf{U}=\mathbf{0}^T$.
  \item Putting \eqref{P1c} and \eqref{P1d} into the objective function of \eqref{Plb}, the inequality constraints \eqref{P1c}--\eqref{P1d} can be removed.
\end{itemize}
Based on the above procedure, problem \eqref{Plb} is equivalently reformulated into
\begin{align}
\mathrm{P}2:\mathop{\mathrm{max}}_{\substack{\mathbf{z},\mathbf{w}}}
~&
\mathrm{min}\left[
\mathop{\mathrm{min}}_{k\neq K}~\Phi_k(\mathbf{z},\mathbf{w}),
\Upsilon(\mathbf{w})
\right],
\nonumber\\
~~~\mathrm{s. t.}~~&||\mathbf{z}||^2_2+||\mathbf{w}||^2_2\leq 2P_{\rm{max}}, \label{P2a}
\end{align}
where
\begin{align}
\Phi_k(\mathbf{z},\mathbf{w})&=
\frac{1}{2}\mathrm{log}_2\left(1+\frac{|\mathbf{g}_k^H\mathbf{U}\mathbf{z}|^2}{\sigma_k^2}\right)
\nonumber\\
&\quad{}
-
\frac{1}{2}\mathrm{log}_2
\left(1+\frac{\lambda\mathbf{w}^H\mathbf{J}\mathbf{w}}{\sigma^2_E}\right), \label{Phi}
\\
\Upsilon(\mathbf{w})&=
\frac{1}{2}\mathrm{log}_2\left(1+\frac{|\mathbf{h}_K^H\mathbf{w}|^2}{\sigma_K^2}\right)
\nonumber\\
&\quad{}
-\frac{1}{2}\mathrm{log}_2
\left(1+\frac{\lambda\mathbf{w}^H\mathbf{J}\mathbf{w}}{\sigma^2_E}\right). \label{Upsilon}
\end{align}

Due to the non-concave functions $\Phi_k$ and $\Upsilon$, problem $\mathrm{P}2$ is nonconvex.
To address this nonconvexity, we will propose a successive convex optimization (SCO) algorithm \cite{sco1,sco2,sco3,sco4,sco5,sco6}, which constructs a sequence of lower bounds on $\{\Phi_k,\Upsilon\}$ and solves the surrogate problems.
There are many techniques to construct the surrogate functions, e.g., convexity inequality, second-order Taylor expansion, arithmetic-geometric mean inequality, and Cauchy-Schwartz inequality \cite{sco1}.
In this paper, the surrogate functions are found by first-order Taylor expansion.

More specifically, given any feasible solution $(\mathbf{z}^\star,\mathbf{w}^\star)$ to $\mathrm{P}2$, we define surrogate functions
\begin{align}
&\widetilde{\Phi}_k(\mathbf{z},\mathbf{w}|\mathbf{z}^\star,\mathbf{w}^\star)
\nonumber\\
&=
\frac{1}{2}\mathrm{log}_2\left(1+2\mathrm{Re}\left[\frac{(\mathbf{z}^\star)^H\mathbf{U}^H\mathbf{g}_{k}\mathbf{g}^H_{k}\mathbf{U}\mathbf{z}}{\sigma_k^2}\right]
-\frac{|\mathbf{g}^H_{k}\mathbf{U}\mathbf{z}^\star|^2}{\sigma_k^2}\right)
\nonumber\\
&~~~-\frac{\lambda\mathbf{w}^H\mathbf{J}\mathbf{w}-\lambda\left(\mathbf{w}^\star\right)^H\mathbf{J}\mathbf{w}^\star}
{2\mathrm{ln}2\left[\sigma_E^2+\lambda\left(\mathbf{w}^\star\right)^H\mathbf{J}\mathbf{w}^\star\right]}
\nonumber\\
&~~~
-\frac{1}{2}\mathrm{log}_2
\left(1+\frac{\lambda\left(\mathbf{w}^\star\right)^H\mathbf{J}\mathbf{w}^\star}{\sigma^2_E}\right), \label{Phiapp}
\\
&\widetilde{\Upsilon}(\mathbf{w}|\mathbf{w}^\star)
\nonumber\\
&=
\frac{1}{2}
\mathrm{log}_2\left(1+2\mathrm{Re}\left[\frac{(\mathbf{w}^\star)^H\mathbf{h}_K\mathbf{h}_K^H\mathbf{w}}{\sigma_K^2}\right]
-\frac{|\mathbf{h}_K^H\mathbf{w}^\star|^2}{\sigma_K^2}
\right)
\nonumber\\
&~~~-\frac{\lambda\mathbf{w}^H\mathbf{J}\mathbf{w}-\lambda\left(\mathbf{w}^\star\right)^H\mathbf{J}\mathbf{w}^\star}
{2\mathrm{ln}2\left[\sigma_E^2+\lambda\left(\mathbf{w}^\star\right)^H\mathbf{J}\mathbf{w}^\star\right]}
\nonumber\\
&~~~
-\frac{1}{2}\mathrm{log}_2
\left(1+\frac{\lambda\left(\mathbf{w}^\star\right)^H\mathbf{J}\mathbf{w}^\star}{\sigma^2_E}\right), \label{Upsilonapp}
\end{align}
and the following proposition can be established.
\begin{proposition}
The functions $\{\widetilde{\Phi}_k,\widetilde{\Upsilon}\}$ satisfy the following properties: (i) $\widetilde{\Phi}_k(\mathbf{z},\mathbf{w}|\mathbf{z}^\star,\mathbf{w}^\star)\leq \Phi_k(\mathbf{z},\mathbf{w})$ and $\widetilde{\Upsilon}(\mathbf{w}|\mathbf{w}^\star)\leq \Upsilon(\mathbf{w})$; (ii) $\widetilde{\Phi}_k(\mathbf{z}^\star,\mathbf{w}^\star|\mathbf{z}^\star,\mathbf{w}^\star)=\Phi_k(\mathbf{z}^\star,\mathbf{w}^\star)$ and
$\widetilde{\Upsilon}(\mathbf{w}^\star|\mathbf{w}^\star)=\Upsilon(\mathbf{w}^\star)$; (iii) With complex gradient operator $\nabla_{\mathbf{x}}:=\partial /\partial\, \mathrm{conj}(\mathbf{x})$,
\begin{align}
\nabla_{\mathbf{z}}\widetilde{\Phi}_k(\mathbf{z},\mathbf{w}|\mathbf{z}^\star,\mathbf{w}^\star)
\Big|_{\mathbf{z}=\mathbf{z}^\star,\mathbf{w}=\mathbf{w}^\star}
=&\nabla_{\mathbf{z}}\Phi_k(\mathbf{z},\mathbf{w})
\Big|_{\mathbf{z}=\mathbf{z}^\star,\mathbf{w}=\mathbf{w}^\star},
\nonumber
\\
\nabla_{\mathbf{w}}\widetilde{\Phi}_k(\mathbf{z},\mathbf{w}|\mathbf{z}^\star,\mathbf{w}^\star)
\Big|_{\mathbf{z}=\mathbf{z}^\star,\mathbf{w}=\mathbf{w}^\star}
=&\nabla_{\mathbf{w}}\Phi_k(\mathbf{z},\mathbf{w})
\Big|_{\mathbf{z}=\mathbf{z}^\star,\mathbf{w}=\mathbf{w}^\star},
\nonumber\\
\nabla_{\mathbf{w}}\widetilde{\Upsilon}(\mathbf{w}|\mathbf{w}^\star)
\Big|_{\mathbf{w}=\mathbf{w}^\star}
=&\nabla_{\mathbf{w}}\Upsilon(\mathbf{w})
\Big|_{\mathbf{w}=\mathbf{w}^\star}.
\nonumber
\end{align}
\end{proposition}
\begin{proof}
See Appendix B.
\end{proof}

With the above observation, a lower bound can be directly obtained if we replace the functions $\{\Phi_k,\Upsilon\}$ by $\{\widetilde{\Phi}_k,\widetilde{\Upsilon}\}$ expanded around a feasible point.
However, a tighter lower bound can be achieved if we treat the obtained solution as another feasible point and continue to construct the next-round surrogate function.
In particular, assuming that the solution at the $n^{\mathrm{th}}$ iteration is given by $(\mathbf{z}^{[n]},\mathbf{w}^{[n]})$, the proposed SCO executes the following two steps at the $(n+1)^{\mathrm{th}}$ iteration:
\begin{itemize}
  \item Use CVX Mosek \cite{opt2}, a Matlab software package for convex optimization, to solve
  \begin{align}
\mathrm{P}2[n+1]:\mathop{\mathrm{max}}_{\substack{\mathbf{z},\mathbf{w}}}
~&
\mathrm{min}\Big[
\mathop{\mathrm{min}}_{k\neq K}~\widetilde{\Phi}_k(\mathbf{z},\mathbf{w}|\mathbf{z}^{[n]},\mathbf{w}^{[n]}),
\nonumber\\
&~~~~~~
\widetilde{\Upsilon}(\mathbf{w}|\mathbf{w}^{[n]})
\Big],
\nonumber\\
~~~\mathrm{s. t.}~~&||\mathbf{z}||^2_2+||\mathbf{w}||^2_2\leq 2P_{\rm{max}}. \label{P2n}
\end{align}
  \item Denoting the optimal solution to $\mathrm{P}2[n+1]$ as
$(\mathbf{z}^*,\mathbf{w}^*)$, we set $(\mathbf{z}^{[n+1]},\mathbf{w}^{[n+1]})=(\mathbf{z}^*,\mathbf{w}^*)$, and the process repeats with solving the problem $\mathrm{P}2[n+2]$.
\end{itemize}

According to \textbf{Proposition 1} and \cite[Theorem 1]{sco1}, the sequence $\{(\mathbf{z}^{[0]},\mathbf{w}^{[0]}),(\mathbf{z}^{[1]},\mathbf{w}^{[1]}),\cdots\}$ converges to the KKT solution to $\mathrm{P}2$ for any feasible starting point $(\mathbf{z}^{[0]},\mathbf{w}^{[0]})$.
In terms of computational complexity, $\mathrm{P}2[n+1]$ involves $N+K-2$ variables, and the complexity for solving $\mathrm{P}2[n+1]$ can be computed to be
$\mathcal{O}[(N+K-2)^{3.5}]$ \cite{opt1}.
As a consequence, the total complexity for solving $\mathrm{P}2$ is $\mathcal{O}[\mathcal{I}(N+K-2)^{3.5}]$, where $\mathcal{I}$ is the number of iterations needed for the SCO based AUCC algorithm to converge.

\section{Large-Scale Optimization Algorithms}

While a KKT solution to $\mathrm{P}2$ has been obtained in Section IV, the corresponding SCO algorithm requires a cubic complexity with respect to $N$ and $K$.
This leads to extremely time-consuming computations in the case of large $N$ or $K$. 
Therefore, efficient large-scale optimization algorithms become indispensable.
Below, we will sequentially consider the cases of large $N$, large $K$ and large $(N,K)$.

\subsection{Large Number of Antennas}

When $N$ is very large such that $N\gg K$, the channels from BS to users would be approximately orthogonal and we have $|\mathbf{g}_j^H\mathbf{g}_k|^2/||\mathbf{g}_k||_2^2\rightarrow 0$ for any $j\neq k$ \cite{massive2}.
Based on such orthogonal feature, the following proposition can be established.
\begin{proposition}
The asymptotic optimal $\mathbf{z}^*$ and $\mathbf{w}^*$ to $\mathrm{P}2$ when $N\rightarrow+\infty$ satisfy
\begin{align}\label{z*}
&\mathbf{z}^*=\mathop{\sum}_{k=1}^{K-1}\sqrt{\frac{\left(2P_{\rm{max}}-||\mathbf{w}^*||_2^2\right) \sigma_k^2}{\mathop{\sum}_{i=1}^{K-1}\sigma_i^2/||\mathbf{U}^H\mathbf{g}_i||^2_2}}\,
\frac{\mathbf{U}^H\mathbf{g}_k}{||\mathbf{U}^H\mathbf{g}_k||_2^2}.
\end{align}
\end{proposition}
\begin{proof}
See Appendix C.
\end{proof}
Using the result from \textbf{Proposition 2}, we can eliminate the vector $\mathbf{z}$ in $\mathrm{P}2$ by setting $\mathbf{z}=\mathbf{z}^*$, which would not change the optimal objective value of $\mathrm{P}2$.
After the above procedure, $\Phi_k$ in the objective function becomes
\begin{align}
\Phi_1(\mathbf{z}^*, \mathbf{w})=\cdots=&\Phi_{K-1}(\mathbf{z}^*,\mathbf{w})
=
\Psi(\mathbf{w}), \nonumber
\end{align}
where
\begin{align}
\Psi(\mathbf{w})=&\frac{1}{2}\mathrm{log}_2\left(1+\frac{2P_{\rm{max}}-||\mathbf{w}||_2^2}{\mathop{\sum}_{i=1}^{K-1}\sigma_i^2/||\mathbf{U}^H\mathbf{g}_i||^2_2}\right)
\nonumber\\
&
-\frac{1}{2}\mathrm{log}_2
\left(1+\frac{\lambda\mathbf{w}^H\mathbf{J}\mathbf{w}}{\sigma^2_E}\right).
\label{Psi}
\end{align}
Moreover, since the constraint \eqref{P2a} is always satisfied by $\mathbf{z}$ in \eqref{z*}, it can be removed accordingly.

Based on the above transformation, problem $\mathrm{P}2$ under $N\rightarrow+\infty$ is equivalent to
\begin{align}
\mathrm{P}3:\mathop{\mathrm{max}}_{\substack{\mathbf{w}}}
~&
\mathrm{min}\left[\Psi(\mathbf{w}),\Upsilon(\mathbf{w})
\right].
\label{P3}
\end{align}
Now, it can be seen that problem $\mathrm{P}3$ only involves one variable $\mathbf{w}$, and is of much lower dimension than $\mathrm{P}2$.
Although $\mathrm{P}3$ is a nonconvex problem, the following proposition is established, which derives the closed-form solution of $\mathbf{w}$.

\begin{proposition}
The optimal $\mathbf{w}^*$ to $\mathrm{P}3$ is
\begin{align}\label{w*}
&\mathbf{w}^*=\frac{\left(\mathbf{\Xi}+\lambda\mathbf{J}/\sigma_E^2\right)^{-1/2}\mathbf{q}}
{\sqrt{\mathbf{q}^H\left(\mathbf{\Xi}+\lambda\mathbf{J}/\sigma_E^2\right)^{-1/2}\mathbf{\Xi}\left(\mathbf{\Xi}+\lambda\mathbf{J}/\sigma_E^2\right)^{-1/2}\mathbf{q}}},
\end{align}
where
\begin{align}\label{Xi}
&\mathbf{\Xi}=\frac{1}{2P_{\rm{max}}}\left[\frac{\mathbf{h}_K\mathbf{h}_K^H}
{\sigma_K^2}\left(\mathop{\sum}_{i=1}^{K-1}\frac{\sigma_i^2}{||\mathbf{U}^H\mathbf{g}_i||^2_2}\right) +\mathbf{I}\right],
\end{align}
and $\mathbf{q}$ is the dominant eigenvector of
\begin{align}
\left(\mathbf{\Xi}+\lambda\mathbf{J}/\sigma_E^2\right)^{-1/2}\left(\mathbf{\Xi}+\mathbf{h}_K\mathbf{h}_K^H/\sigma_K^2\right)
\left(\mathbf{\Xi}+\lambda\mathbf{J}/\sigma_E^2\right)^{-1/2}. \nonumber
\end{align}
\end{proposition}
\begin{proof}
See Appendix D.
\end{proof}
Since $\mathrm{P}3$ is equivalent to $\mathrm{P}2$ under $N\rightarrow\infty$,
the asymptotic optimal $\mathbf{w}^*$ to $\mathrm{P}2$ is given by \eqref{w*} according to \textbf{Proposition 3}.
By putting $\mathbf{w}=\mathbf{w}^*$ into \eqref{z*} of \textbf{Proposition 2}, the optimal $\mathbf{z}^*$ to $\mathrm{P}2$ under $N\rightarrow\infty$ is also obtained.
Therefore, the entire procedure for deriving the solution to $\mathrm{P}2$ with large $N$ is to execute \eqref{w*} and \eqref{z*} sequentially.
The computational complexity of the closed-form AAUC algorithm is dominated by computing $\left(\mathbf{\Xi}+\lambda\mathbf{J}/\sigma_E^2\right)^{-1/2}$, which requires a complexity of $\mathcal{O}((K-1)^{3})$.

\subsection{Large Number of Users}

When $K$ is very large, the orthogonal condition does not hold and the closed-form AAUC scheme is not applicable.
To this end, this subsection will propose the BFOM based AAUC scheme for solving $\mathrm{P}2$ when $K$ is large.

More specifically, as $K\rightarrow+\infty$, $||\mathbf{h}_K||_2^2\gg \mathrm{Tr}(\mathbf{J})$ holds. 
Therefore, we have $\Upsilon(\mathbf{w})\rightarrow \frac{1}{2}\mathrm{log}_2\left(1+|\mathbf{h}_K^H\mathbf{w}|^2/\sigma_K^2\right)$, and the optimal $\mathbf{w}^*$ is given by $\mathbf{w}^*=\sqrt{p}\,\mathbf{h}_K/||\mathbf{h}_K||_2$, with $p\geq 0$ being the transmit power to be determined.
By putting $\mathbf{w}=\sqrt{p}\,\mathbf{h}_K/||\mathbf{h}_K||_2$ into $\mathrm{P}2$, $\mathrm{P}2$ under $K\rightarrow+\infty$ is re-written as
\begin{align}
\mathrm{P}4:\mathop{\mathrm{max}}_{\substack{\mathbf{z},\,p\geq 0}}
~&\mathop{\mathrm{min}}\left(\mathop{\mathrm{min}}_{k=1,...,K-1}~
\frac{|\mathbf{g}^H_{k}\mathbf{U}\mathbf{z}|^2}{\sigma_k^2},\,\frac{p||\mathbf{h}_K||_2^2}{\sigma_K^2}\right),
\nonumber\\
~~~~~~\mathrm{s. t.}~~&||\mathbf{z}||^2_2+p\leq 2P_{\rm{max}}, 
\end{align}
where we have removed $\frac{1}{2}\mathrm{log}_2(1+\cdot)$ in the objective function due to its monotonicity.
It can be seen from $\mathrm{P}4$ that the objective function to be maximized is non-concave.
Therefore, $\mathrm{P}4$ is a nonconvex optimization problem, which can also be solved via the SCO algorithm.
Specifically, with any feasible starting point $(\mathbf{z}^{[0]},p^{[0]})$, we execute the following update at the $n^{\mathrm{th}}$ iteration:
\begin{align}
&\left(\mathbf{z}^{[n+1]},p^{[n+1]}\right)
=\mathop{\mathrm{argmax}}_{\substack{\mathbf{z},\,p\geq 0\\||\mathbf{z}||^2_2+p\leq 2P_{\rm{max}}}}
\mathop{\mathrm{min}}\Bigg\{\mathop{\mathrm{min}}_{k=1,...,K-1}~
\nonumber\\
&
\left(2\mathrm{Re}\left[\frac{(\mathbf{z}^{[n]})^H\mathbf{U}^H\mathbf{g}_{k}\mathbf{g}^H_{k}\mathbf{U}\mathbf{z}}{\sigma_k^2}\right)-\frac{|\mathbf{g}^H_{k}\mathbf{U}\mathbf{z}^{[n]}|^2}{\sigma_k^2}\right),
\,\frac{p||\mathbf{h}_K||_2^2}{\sigma_K^2}\Bigg\}, \label{P4n+1}
\end{align}
According to \cite{sco1,sco2,sco3,sco4,sco5,sco6}, the above iterative procedure is guaranteed to converge to a KKT solution to $\mathrm{P}4$. 

Based on the SCO framework, the remaining question is how to compute \eqref{P4n+1}. 
The major challenge comes from the non-smooth operator $\mathop{\mathrm{min}}_{k=1,...,K-1}$ in the objective function, which hinders us from computing the gradients.
To this end, in the following, problem \eqref{P4n+1} will be reformulated into a smooth bilevel optimization problem with $\ell_1$ norm and $\ell_2$ norm constraints.
Furthermore, since the projection onto $\ell_1$ norm in Euclidean space involves a high computational complexity, we propose to execute the projection in non-Euclidean space.
This is achieved by using Kullback-Leibler (KL) divergence as the distance measure, and the resultant algorithm is termed BFOM based AAUC.

First of all, we derive the procedure that equivalently transforms \eqref{P4n+1} into a smooth optimization problem.
To this end, variables $\mathbf{x}=[\mathrm{Re}(\mathbf{z}^T),\mathrm{Im}(\mathbf{z}^T)]^T\in\mathbb{R}^{(2N-2)\times 1}$ and $\bm{\gamma}\in\mathbb{R}^{K\times 1}$ are introduced with
$\mathbf{x}\in\mathcal{A}$ and $\bm{\gamma}\in\mathcal{B}$, where
\begin{align}
\mathcal{A}=&\left\{\mathbf{x}:||\mathbf{x}||^2_2\leq 2P_{\rm{max}}\right\},
\nonumber\\
\mathcal{B}=&\{\bm{\gamma}\in\mathbb{R}^{K\times 1}_+:||\bm{\gamma}||_1=1\}.
\end{align}
Then \eqref{P4n+1} is equivalent to
\begin{align}
&\mathrm{P}5:
\mathop{\mathrm{max}}_{\substack{\mathbf{x}\in\mathcal{A}}}
~\mathop{\mathrm{min}}_{\bm{\gamma}\in\mathcal{B}}~\Theta^{[n]}\left(\mathbf{x},\bm{\gamma}\right),
\end{align}
where
\begin{align}\label{rkn}
\Theta^{[n]}\left(\mathbf{x},\bm{\gamma}\right)&=\mathop{\sum}_{k=1}^{K-1}\gamma_k\left((\mathbf{r}_k^{[n]})^T\mathbf{x}
+t_k^{[n]}\right),
\nonumber\\
&~~~
+\frac{\gamma_K||\mathbf{h}_K||_2^2\left(2P_{\rm{max}}-||\mathbf{x}||_2^2\right)}{\sigma_K^2},
\nonumber\\
\mathbf{r}_{k}^{[n]}&=\frac{2}{\sigma_k^2}
\left[
\begin{array}{cccc}
\mathrm{Re}\left(\mathbf{U}^H\mathbf{g}_{k}\mathbf{g}_{k}^H\mathbf{U}\mathbf{z}^{[n]}\right)
\\
\mathrm{Im}\left(\mathbf{U}^H\mathbf{g}_{k}\mathbf{g}_{k}^H\mathbf{U}\mathbf{z}^{[n]}\right)
\end{array}
\right],\quad \forall k\neq K,
\nonumber\\
t_k^{[n]}&=-\frac{|\mathbf{g}^H_{k}\mathbf{U}\mathbf{z}^{[n]}|^2}{\sigma_k^2},\quad \forall k\neq K.
\end{align}

Now, it can be seen from $\mathrm{P}5$ that $\Theta^{[n]}(\mathbf{x},\bm{\gamma})$ is differentiable in both $\mathbf{x}$ and $\bm{\gamma}$.
Moreover, $\Theta^{[n]}$ is a bilevel function of two variables, with the upper layer variable $\mathbf{x}$ constrained by $\ell_2$ norm and the lower layer variable $\bm{\gamma}$ constrained by $\ell_1$ norm.
To address this bilevel problem, the proposed BFOM starts from a feasible $\mathbf{x}=\mathbf{y}^{[0]}\in\mathcal{A}$ and $\bm{\gamma}=\bm{\eta}^{[0]}\in\mathcal{B}$ (e.g., $\mathbf{y}^{[0]}=\bm{0}$ and $\bm{\eta}^{[0]}=\bm{1}_K/K$), and solves $\mathrm{P}5$ via the following update at the $(m+1)^{\mathrm{th}}$ iteration \cite{BFOM1}:
\begin{subequations}
\begin{align}
\mathbf{y}^{\diamond}&=
\mathop{\mathrm{argmin}}_{\mathbf{x}\in\mathcal{A}}~
W_{\mathcal{A}}\left(\mathbf{x},\mathbf{y}^{[m]}\right)
\nonumber\\
&\quad{}
-\frac{1}{2L}\,\mathbf{x}^T\nabla_{\mathbf{x}}\Theta^{[n]}
\left(\mathbf{x},\bm{\gamma}\right)|_{\mathbf{x}=\mathbf{y}^{[m]},\bm{\gamma}=\bm{\eta}^{[m]}},
\label{BFOM1}
\\
\bm{\eta}^{\diamond}&=
\mathop{\mathrm{argmin}}_{\bm{\gamma}\in\mathcal{B}}~W_{\mathcal{B}}\left(\bm{\gamma},\bm{\eta}^{[m]}\right)
\nonumber\\
&\quad{}
+\frac{1}{2L}\,\bm{\gamma}^T\nabla_{\bm{\gamma}}\Theta^{[n]}
\left(\mathbf{x},\bm{\gamma}\right)|_{\mathbf{x}=\mathbf{y}^{[m]},\bm{\gamma}=\bm{\eta}^{[m]}}, \label{BFOM2}
\\
\mathbf{y}^{[m+1]}&=
\mathop{\mathrm{argmin}}_{\mathbf{x}\in\mathcal{A}}~
W_{\mathcal{A}}\left(\mathbf{x},\mathbf{y}^{[m]}\right)
\nonumber\\
&\quad{}
-\frac{1}{2L}\, \mathbf{x}^T\nabla_{\mathbf{x}}\Theta^{[n]}
\left(\mathbf{x},\bm{\gamma}\right)|_{\mathbf{x}=\mathbf{y}^{\diamond},\bm{\gamma}=\bm{\eta}^{\diamond}},
\label{BFOM3}
\\
\bm{\eta}^{[m+1]}&=
\mathop{\mathrm{argmin}}_{\bm{\gamma}\in\mathcal{B}}~W_{\mathcal{B}}\left(\bm{\gamma},\bm{\eta}^{[m]}\right)
\nonumber\\
&\quad{}
+\frac{1}{2L}\,\bm{\gamma}^T\nabla_{\bm{\gamma}}\Theta^{[n]}
\left(\mathbf{x},\bm{\gamma}\right)|_{\mathbf{x}=\mathbf{y}^{\diamond},\bm{\gamma}=\bm{\eta}^{\diamond}}, \label{BFOM4}
\end{align}
\end{subequations}
where $W_{\mathcal{A}}$ is the Euclidean distance induced by $\ell_2$ norm and $W_{\mathcal{B}}$ is the KL divergence induced by $\ell_1$ norm:
\begin{align}
&W_{\mathcal{A}}\left(\mathbf{x},\mathbf{y}^{[m]}\right)=\frac{1}{2}||\mathbf{x}-\mathbf{y}^{[m]}||^2_2,  \label{DA}
\\
&W_{\mathcal{B}}\left(\bm{\gamma},\bm{\eta}^{[m]}\right)=\mathop{\sum}_{k=1}^K\gamma_k\,\mathrm{ln}\left(\frac{\gamma_k}{\eta^{[m]}_k}\right).  \label{DB}
\end{align}
Furthermore, the gradients $\nabla_{\mathbf{x}}\Theta^{[n]}$ and $\nabla_{\bm{\gamma}}\Theta^{[n]}$ are given by
\begin{subequations}
\begin{align}
\nabla_{\mathbf{x}}\Theta^{[n]}
\left(\mathbf{x},\bm{\gamma}\right)
=&
\mathop{\sum}_{k=1}^{K-1}\gamma_k\mathbf{r}_k^{[n]}
-\frac{2\gamma_K||\mathbf{h}_K||_2^2\,\mathbf{x}}{\sigma_K^2},
\label{gradient1}\\
\nabla_{\bm{\gamma}}\Theta^{[n]}
\left(\mathbf{x},\bm{\gamma}\right)
=&\Big[(\mathbf{r}_1^{[n]})^T\mathbf{x}+t_1^{[n]},
\cdots,
(\mathbf{r}_{K-1}^{[n]})^T\mathbf{x}+t_{K-1}^{[n]},
\nonumber\\
&
||\mathbf{h}_K||_2^2\left(2P_{\rm{max}}-||\mathbf{x}||_2^2\right)\big/\sigma_K^2\Big]^T.
\label{gradient2}
\end{align}
\end{subequations}
Finally, $L$ is the Bregman Lipschitz constant satisfying
\begin{subequations}
\begin{align}
&||\nabla_{\mathbf{x}}\Theta^{[n]}
\left(\mathbf{x},\bm{\gamma}\right)-\nabla_{\mathbf{x}}\Theta^{[n]}
\left(\mathbf{x}',\bm{\gamma}\right)||_2
\leq
L||\mathbf{x}-\mathbf{x}'||_2,
\label{Lips1}\\
&||\nabla_{\mathbf{x}}\Theta^{[n]}
\left(\mathbf{x},\bm{\gamma}\right)-\nabla_{\mathbf{x}}\Theta^{[n]}
\left(\mathbf{x},\bm{\gamma}'\right)||_2
\leq L||\bm{\gamma}-\bm{\gamma}'||_1,
\label{Lips2} \\
&||\nabla_{\bm{\gamma}}\Theta^{[n]}
\left(\mathbf{x},\bm{\gamma}\right)-\nabla_{\bm{\gamma}}\Theta^{[n]}
\left(\mathbf{x},\bm{\gamma}'\right)||_{\infty}
\leq
L||\bm{\gamma}-\bm{\gamma}'||_1,
\label{Lips3} \\
&||\nabla_{\bm{\gamma}}\Theta^{[n]}
\left(\mathbf{x},\bm{\gamma}\right)-\nabla_{\bm{\gamma}}
\Theta^{[n]}
\left(\mathbf{x}',\bm{\gamma}\right)||_{\infty}
\leq L||\mathbf{x}-\mathbf{x}'||_2, \label{Lips4} \\
&\forall\mathbf{x}, \mathbf{x}'\in\mathcal{A},\quad \forall\bm{\gamma}, \bm{\gamma}'\in\mathcal{B}. \nonumber
\end{align}
\end{subequations}
The following proposition derives a valid $L$ satisfying the above conditions.
\begin{proposition}
The conditions \eqref{Lips1}--\eqref{Lips4} are satisfied with $L=\widehat{L}$, where
\begin{align}
\widehat{L}=&\mathrm{max}~\Big[||\mathbf{r}_1^{[n]}||_2,\cdots,||\mathbf{r}_{K-1}^{[n]}||_2,
\nonumber\\
&
2||\mathbf{h}_K||_2^2\,\mathrm{max}\left(\sqrt{2P_{\rm{max}}},1\right)\big/\sigma_K^2
\Big]. \label{Lips}
\end{align}
\end{proposition}
\begin{proof}
See Appendix E.
\end{proof}
Although $\widehat{L}$ is a valid Lipschiz constant according to \textbf{Proposition 4}, whether a smaller $L$ with $L<\widehat{L}$ would satisfy \eqref{Lips1}--\eqref{Lips4} is not known.
Thus, in practice we need to fine-tune the hyper-parameter $L$ and this paper sets $L=\widehat{L}/K$.

The key idea of \eqref{BFOM1}--\eqref{BFOM4} is to update the variables along their gradient direction, while keeping the updated point $\{\mathbf{y}^{\diamond},\bm{\eta}^{\diamond},\mathbf{y}^{[m+1]},\bm{\eta}^{[m+1]}\}$ close to the current point $\{\mathbf{y}^{[m]},\bm{\eta}^{[m]}\}$.
Furthermore, by using the intermediate point $\{\mathbf{y}^{\diamond},\bm{\eta}^{\diamond}\}$ in \eqref{BFOM1}--\eqref{BFOM2}, we can compute the look-ahead gradient for updating $\{\mathbf{y}^{[m+1]},\bm{\eta}^{[m+1]}\}$ as in \eqref{BFOM3}--\eqref{BFOM4}.
As proved in \cite{BFOM1,BFOM2,BFOM3,BFOM4}, the iterative procedure \eqref{BFOM1}--\eqref{BFOM4} is guaranteed to converge to the optimal solution to $\mathrm{P}5$ with a convergence rate of $\mathcal{O}(1/m)$.

Notice that the KKT optimality conditions can be adopted to derive the closed-form expressions for \eqref{BFOM1}--\eqref{BFOM4}.
In particular, equations \eqref{BFOM1}--\eqref{BFOM2} are equivalent to
\begin{subequations}
\begin{align}
\mathbf{y}^{\diamond}=&
{\left(\mathrm{max}\left[\Big|\Big|\mathbf{y}^{[m]}+\frac{\nabla_{\mathbf{x}}\Theta^{[n]}
\left(\mathbf{y}^{[m]},\bm{\eta}^{[m]}\right)}{2L}
\Big|\Big|_2,\sqrt{2P_{\rm{max}}}\right]\right)}^{-1}
\nonumber\\
&
\times\sqrt{2P_{\rm{max}}}\left[\mathbf{y}^{[m]}+\frac{\nabla_{\mathbf{x}}\Theta^{[n]}
\left(\mathbf{y}^{[m]},\bm{\eta}^{[m]}\right)}{2L}\right]
,
\nonumber
\\
\eta^{\diamond}_k=&
{\left(
\sum_{i=1}^K\eta^{[m]}_i\mathrm{exp}\left[-\frac{\nabla_{\gamma_i}\Theta^{[n]}
\left(\mathbf{y}^{[m]},\bm{\eta}^{[m]}\right)}{2L}\right]\right)}^{-1}
\nonumber\\
&\times
\eta^{[m]}_k\mathrm{exp}\left[-\frac{\nabla_{\gamma_k}\Theta^{[n]}
\left(\mathbf{y}^{[m]},\bm{\eta}^{[m]}\right)}{2L}
\right]
,~~\forall k.
\nonumber
\end{align}
\end{subequations}
On the other hand, the closed-form expressions for \eqref{BFOM3}--\eqref{BFOM4} can be derived similarly.
In terms of computational complexity, the total complexity of the BFOM based AAUC would be $\mathcal{O}(\mathcal{I}\, MKN)$, where $\mathcal{I}$ and $M$ are the number of iterations for the SCO and BFOM algorithms, respectively.

\emph{Remark 2:} Since our aim is to obtain an approximate solution, we can terminate the iterative procedure when the norm $||\mathbf{y}^{[m+1]}-\mathbf{y}^{[m]}||$ is small enough, e.g., $||\mathbf{y}^{[m]}-\mathbf{y}^{[m-1]}||<10^{-4}$, or the number of iterations reaches $3000$.
Then, the optimal $\mathbf{x}^*$ to $\mathrm{P}5$ is given by $\mathbf{x}^{*}=\sqrt{\tau}\,\mathbf{y}^{[m]}/||\mathbf{y}^{[m]}||_2$, where $\tau$ solves
\begin{align}\label{scale}
\mathop{\mathrm{min}}_{k=1,\cdots,K-1}\frac{\sqrt{\tau}(\mathbf{r}_k^{[n]})^T\mathbf{y}^{[m]}}{||\mathbf{y}^{[m]}||_2}
+t_k^{[n]}
=\frac{||\mathbf{h}_K||_2^2\left(2P_{\rm{max}}-\tau\right)}{\sigma_K^2}.
\end{align}
Accordingly, the optimal $\mathbf{z}^*$ and $p^*$ to problem \eqref{P4n+1} can be recovered as $\mathbf{z}^{*}=[x_1^{*},\cdots,x_{N-1}^{*}]^T+\mathrm{j}\,[x_N^{*},\cdots,x_{2N-2}^{*}]^T$ and $p^*=2P_{\mathrm{max}}-||\mathbf{z}^*||^2_2$.

\subsection{Large Number of Antennas and Users}

When both $N$ and $K$ are very large (with $N/K$ being a large constant), we can combine the derivations in Sections V-A and V-B.
More specifically, based on $|\mathbf{g}_j^H\mathbf{g}_k|^2/||\mathbf{g}_k||_2^2\rightarrow 0$ for $j\neq k$ when $N\rightarrow \infty$ and
$\Upsilon(\mathbf{w})\rightarrow \frac{1}{2}\mathrm{log}_2\left(1+|\mathbf{h}_K^H\mathbf{w}|^2/\sigma_K^2\right)$ when $K\rightarrow \infty$, the problem $\rm{P}2$ under $N,K\rightarrow \infty$ is
\begin{align}
\mathrm{P}6:\mathop{\mathrm{max}}_{\substack{p}}
~&
\mathrm{min}\Bigg[\frac{1}{2}\mathrm{log}_2\left(1+\frac{2P_{\rm{max}}-p}{\mathop{\sum}_{i=1}^{K-1}\sigma_i^2/||\mathbf{U}^H\mathbf{g}_i||^2_2}\right),
\nonumber\\
&
\frac{1}{2}\mathrm{log}_2\left(1+\frac{p||\mathbf{h}_K||_2^2}{\sigma_K^2}\right)
\Bigg], \nonumber
\end{align}
where $p$ is the users' total power to be determined.
In the objective function, the first term inside the minimum function is monotonic decreasing in $p$ while the second term is monotonic increasing in $p$.
Therefore, the optimal $p^*$ to $\rm{P}6$ must satisfy
\begin{align}
&\frac{1}{2}\mathrm{log}_2\left(1+\frac{(2P_{\rm{max}}-p)}{\mathop{\sum}_{i=1}^{K-1}\sigma_i^2/||\mathbf{U}^H\mathbf{g}_i||^2_2}\right)
\nonumber\\
&
=\frac{1}{2}\mathrm{log}_2\left(1+\frac{p||\mathbf{h}_K||_2^2}{\sigma_K^2}\right),
\end{align}
which gives
\begin{align}
& p^* = \frac{2P_{\rm{max}}}{1+\dfrac{||\mathbf{h}_K||_2^2}{\sigma_K^2}\left(\mathop{\sum}_{i=1}^{K-1}\sigma_i^2/||\mathbf{U}^H\mathbf{g}_i||^2_2\right)}.
\end{align}
As a result, the optimal $\mathbf{w}^*$ to $\rm{P}2$ when $K$ and $N$ go to infinity can be recovered as $\mathbf{w}^*=\sqrt{p^*}\mathbf{h}_K/||\mathbf{h}_K||_2$. 
With $\mathbf{w}^*$, the optimal $\mathbf{z}^*$ can be computed using \textbf{Proposition 2}.

\section{IRS Based Scheme}

It can be seen from problem $\rm{P}1$ that the proposed AAUC scheme requires extra power consumption from other devices.
To avoid extra power consumption, the proposed scheme can be integrated with the emerging technology of intelligent reflecting surfaces (IRS) \cite{irs1,irs2,irs3,irs4}.
In particular, we can place the IRS elements at the users $k=1,\cdots,K-1$ (for fair comparison with the non-IRS case).
Then each element of the IRS receives the multi-path signals from the BS, and scatters the combined signal with adjustable amplitude and phase to user $K$.
By employing the IRS technique, the two transmission phases can be combined into one, and the system cost of the proposed scheme can be significantly reduced.

More specifically, the optimization problem of the IRS based scheme is given by
\begin{subequations}
\begin{align}
\mathrm{Q}1:&\mathop{\mathrm{max}}_{\substack{\mathbf{v},\mathbf{u},R}}
\quad\frac{1}{D_K}\int_{0}^{D_K}\mathbb{E}_{\{\phi_k,\delta_k\}}\Bigg\{\Bigg[R-\mathrm{log}_2
\Bigg(1+\frac{1}{\sigma^2_E}
\nonumber\\
&\Big|\mathbf{h}^H_E(\rho,\{\phi_k,\delta_k\})\left(\left[\mathbf{g}_1,\cdots,\mathbf{g}_{K-1}\right]^H\mathbf{v}\circ\mathbf{u}\right)
\Big|^2\Bigg)\Bigg]^+\Bigg\}\mathrm{d}\rho,
\nonumber\\
\mathrm{s. t.}\quad&||\mathbf{v}||^2_2 \leq P_{\rm{max}},\quad |\mathbf{g}^H_{K}\mathbf{v}|=0,
\\
&R\leq\mathrm{log}_2\left(1+\frac{|\mathbf{g}^H_{k}\mathbf{v}|^2}{\sigma_k^2}\right),~~\forall k\neq K,
\\
&R\leq \mathrm{log}_2\Bigg(1
+\frac{\Big|\mathbf{h}_K^H\left(\left[\mathbf{g}_1,\cdots,\mathbf{g}_{K-1}\right]^H\mathbf{v}\circ\mathbf{u}\right)\Big|^2}{\sigma_K^2}\Bigg),
\\
&|u_k|_2\leq 1,\quad \forall k\neq K, \label{IRS}
\end{align}
\end{subequations}
where $\mathbf{u}\in\mathbb{C}^{(K-1)\times 1}$ is the amplitude and phase design at IRS.
By introducing a slack variable $\mathbf{c}=\left[\mathbf{g}_1,\cdots,\mathbf{g}_{K-1}\right]^H\mathbf{v}\circ\mathbf{u}\in\mathbb{C}^{(K-1)\times 1}$ with $|c_k|^2\leq |\mathbf{g}_k^H\mathbf{v}|^2$, and adopting the angular secrecy model in \eqref{smodel}, the problem $\rm{Q}1$ is converted into
\begin{subequations}
\begin{align}
\mathrm{Q}2:\mathop{\mathrm{max}}_{\substack{\mathbf{v},\mathbf{c},R}}
~~&\left[R-\mathrm{log}_2
\left(1+\frac{\lambda\mathbf{c}^H\mathbf{J}\mathbf{c}}{\sigma^2_E}\right)\right]^+,
\nonumber\\
~~~~~~\mathrm{s. t.}~~~~&||\mathbf{v}||^2_2 \leq P_{\rm{max}},\quad |\mathbf{g}^H_{K}\mathbf{v}|=0,
\\
&R\leq\mathrm{log}_2\left(1+\frac{|\mathbf{g}^H_{k}\mathbf{v}|^2}{\sigma_k^2}\right),~~\forall k\neq K,
\\
&R\leq \mathrm{log}_2\left(1+\frac{|\mathbf{h}_K^H\mathbf{c}|^2}{\sigma_K^2}\right),
\\
&|c_k|^2\leq |\mathbf{g}_k^H\mathbf{v}|^2,\quad \forall k\neq K,
\end{align}
\end{subequations}
The above problem has a similar structure as $\mathrm{P}1$ and can be solved by the SCO algorithm.
Notice that the channel estimation of the IRS based scheme is challenging due to the low costs of IRS elements.
Therefore, IRS is not applicable to the case of fast changing channels.

\section{Simulation Results and Discussions}

This section provides simulation results to evaluate the performance of the proposed scheme.
It is assumed that the noise powers $\sigma_1^2=\cdots=\sigma_K^2=\sigma_E^2=-80\,\mathrm{dBm}$ \cite{irs1} (corresponding to power spectral density $-140\,\mathrm{dBm/Hz}$ \cite{wang} with $1\,\mathrm{MHz}$ bandwidth).
The transmit power budget is $P_{\rm{max}}=30\,\mathrm{dBm}$ and the pathloss exponent is set to $\alpha=2.5$.
The pathloss at the distance of $1\,\rm{m}$ is $-40\,\rm{dB}$, which is computed using the 3GPP UMi model \cite{3gpp} and $3.5\,\rm{GHz}$ carrier frequency (i.e., carrier frequency of 5G in China).
We adopt $\lambda=0.64$ for the secrecy model in \eqref{smodel}, and the number of iterations for SCO is set to $\mathcal{I}=20$.

Based on the above settings, we simulate the following massive MIMO system:
\begin{itemize}
\item The BS is located at $(0,0)$.
\item The users are located at $(D_k\,\mathrm{cos}\,\theta_k,D_k\,\mathrm{sin}\,\theta_k)$ with $D_k\sim \mathcal{U}(100,500)$ and
$\theta_k\sim \mathcal{U}(-\pi,\pi)$ for all $k$.
\item The eavesdropper is located at $(D_E\,\mathrm{cos}\,\theta_K,D_E\,\mathrm{sin}\,\theta_K)$ with
$D_E\sim \mathcal{U}(0,D_K)$.
\item All the distances use the unit of $\mathrm{m}$.
\end{itemize}

With $\{D_k,\theta_k,D_E,\theta_E\}$, channels $\{\mathbf{g}_{k},\mathbf{g}_E\}$ are generated according to \eqref{racian}.
The channel $\mathbf{h}_E$ is generated according to \eqref{channel}, where the distance $d_{E,k}$ between the eavesdropper and the user $k$ is computed based on their locations.
The channel $\mathbf{h}_K$ is also generated according to \eqref{channel}, but with $d_{E,k}$ replaced by the distance $d_{K,k}$ between the users $K$ and $k$.
Finally, each point in the figures is obtained by averaging over $50$ simulation runs, with independent channel realizations and locations of users in each run.

\begin{figure*}[!t]
\centering
\subfigure[]{\includegraphics[width=58mm]{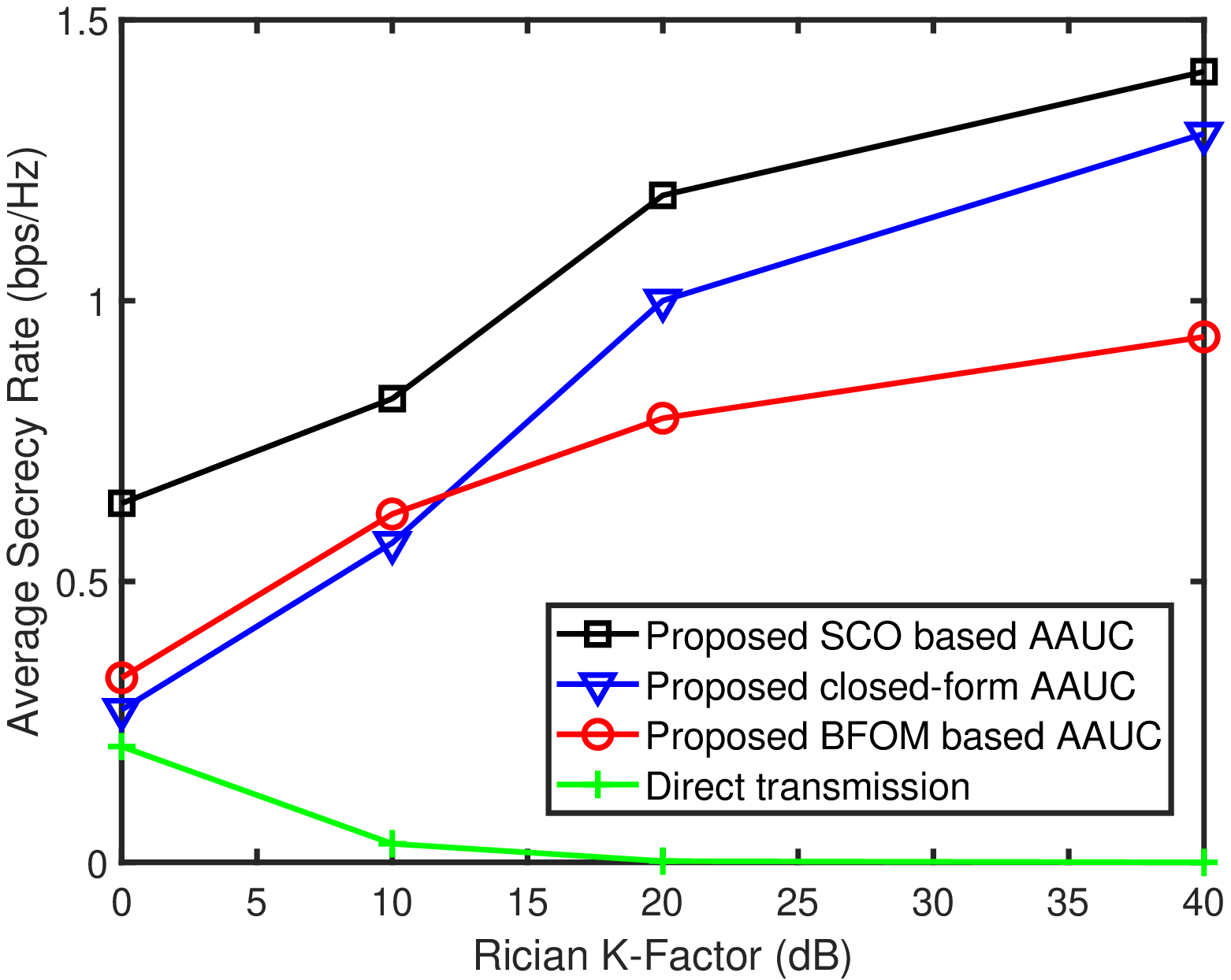}}
\subfigure[]{\includegraphics[width=58mm]{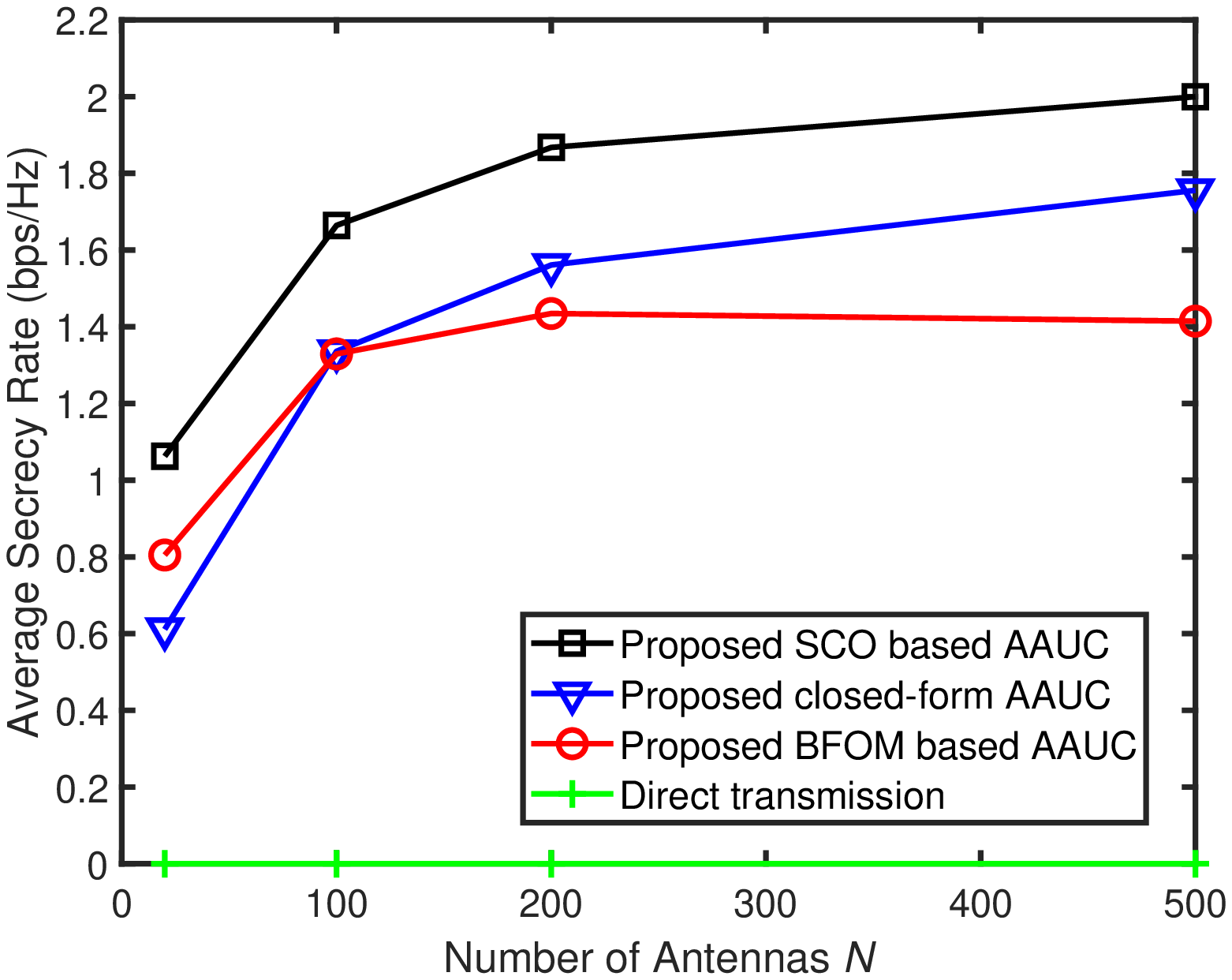}}
\subfigure[]{\includegraphics[width=58mm]{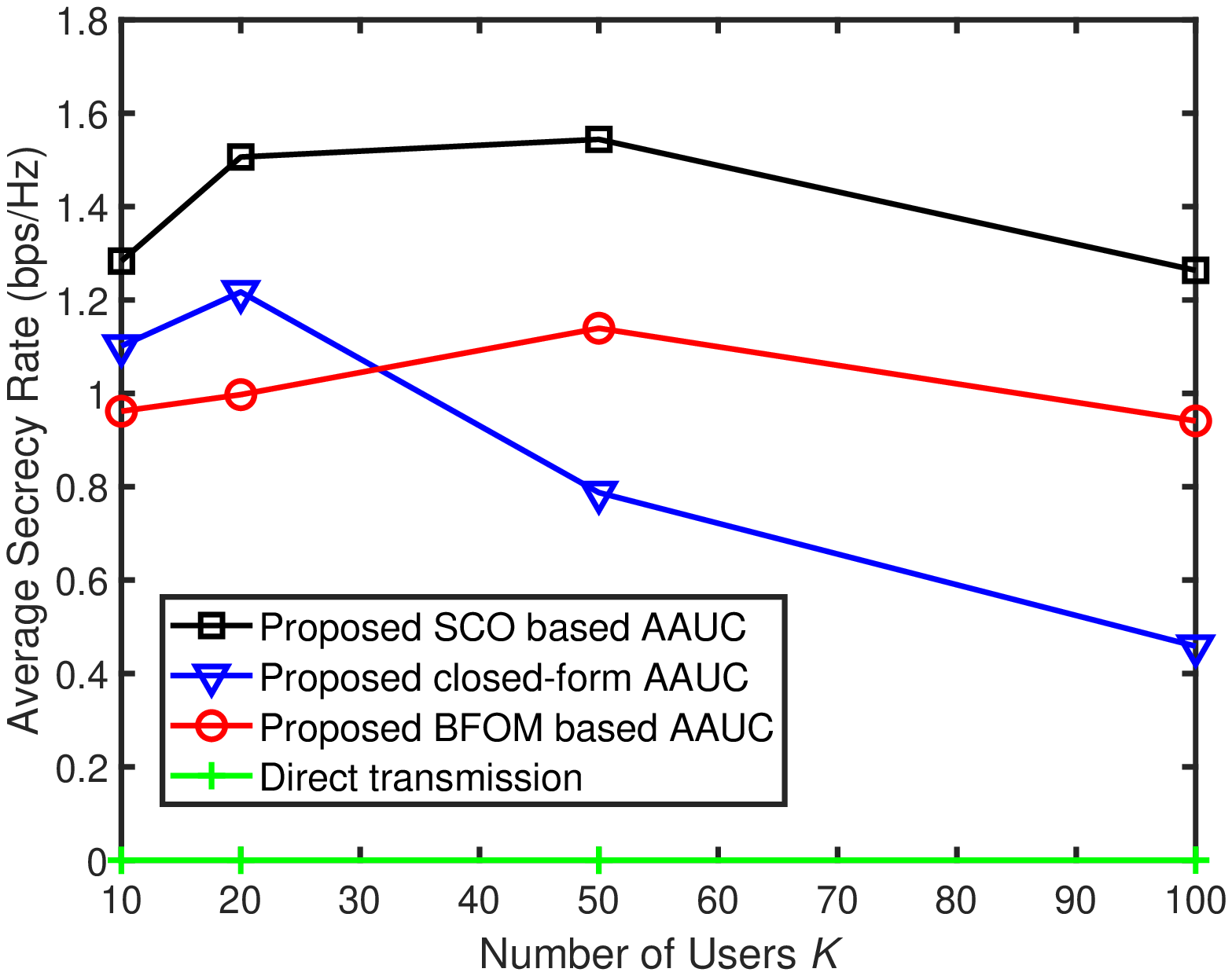}}
\caption{a) Average secrecy rate versus Rician K-factor $K_R$ when $K=10$ and $N=100$; b) Average secrecy rate versus the number of antennas $N$ when $K=20$ and $K_R=30~\mathrm{dB}$; c) Average secrecy rate versus the number of users $K$ when $N=100$ and $K_R=30\,\mathrm{dB}$.}
\end{figure*}

\begin{figure*}[!t]
\centering
\subfigure[]{\includegraphics[width=58mm]{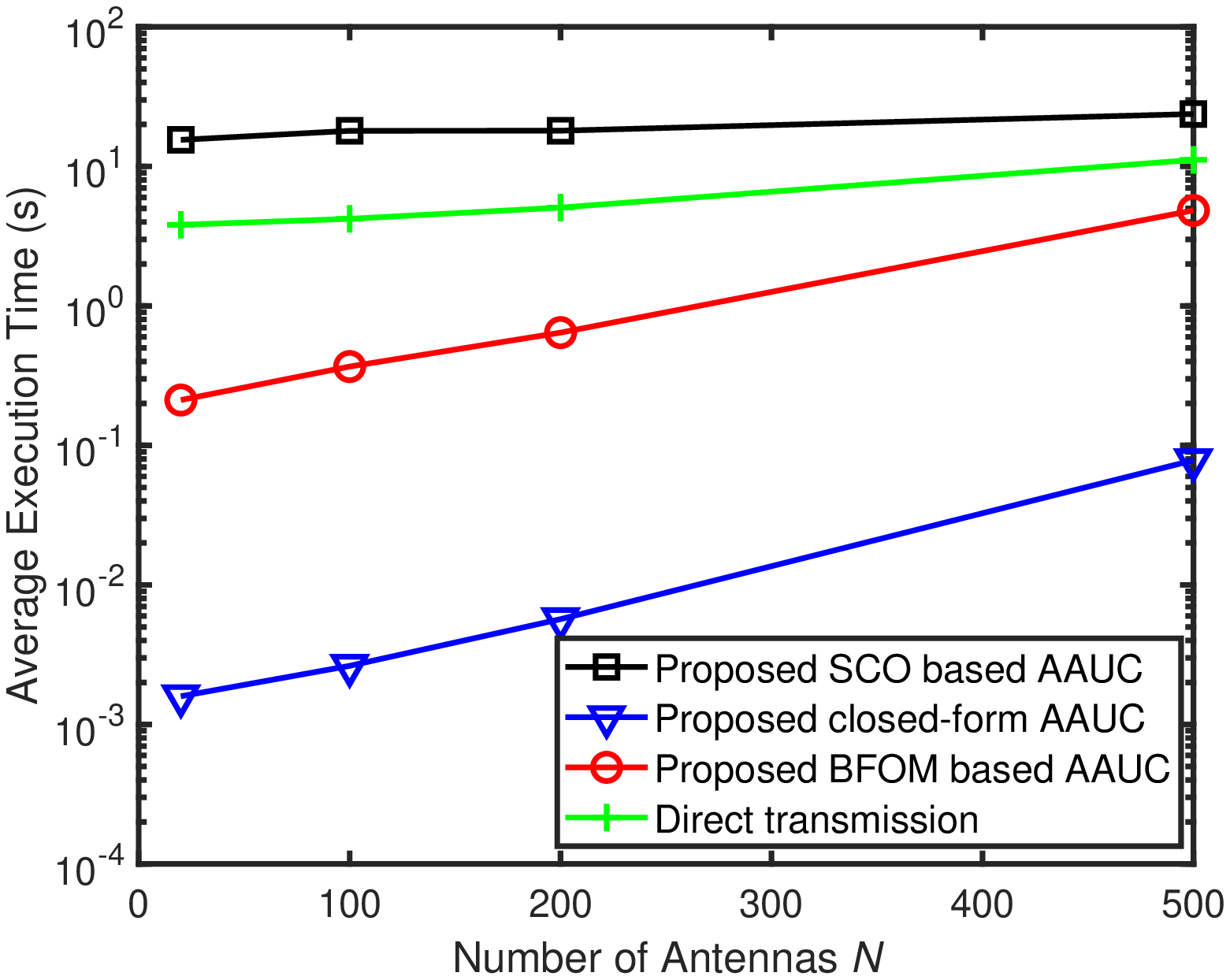}}
\subfigure[]{\includegraphics[width=58mm]{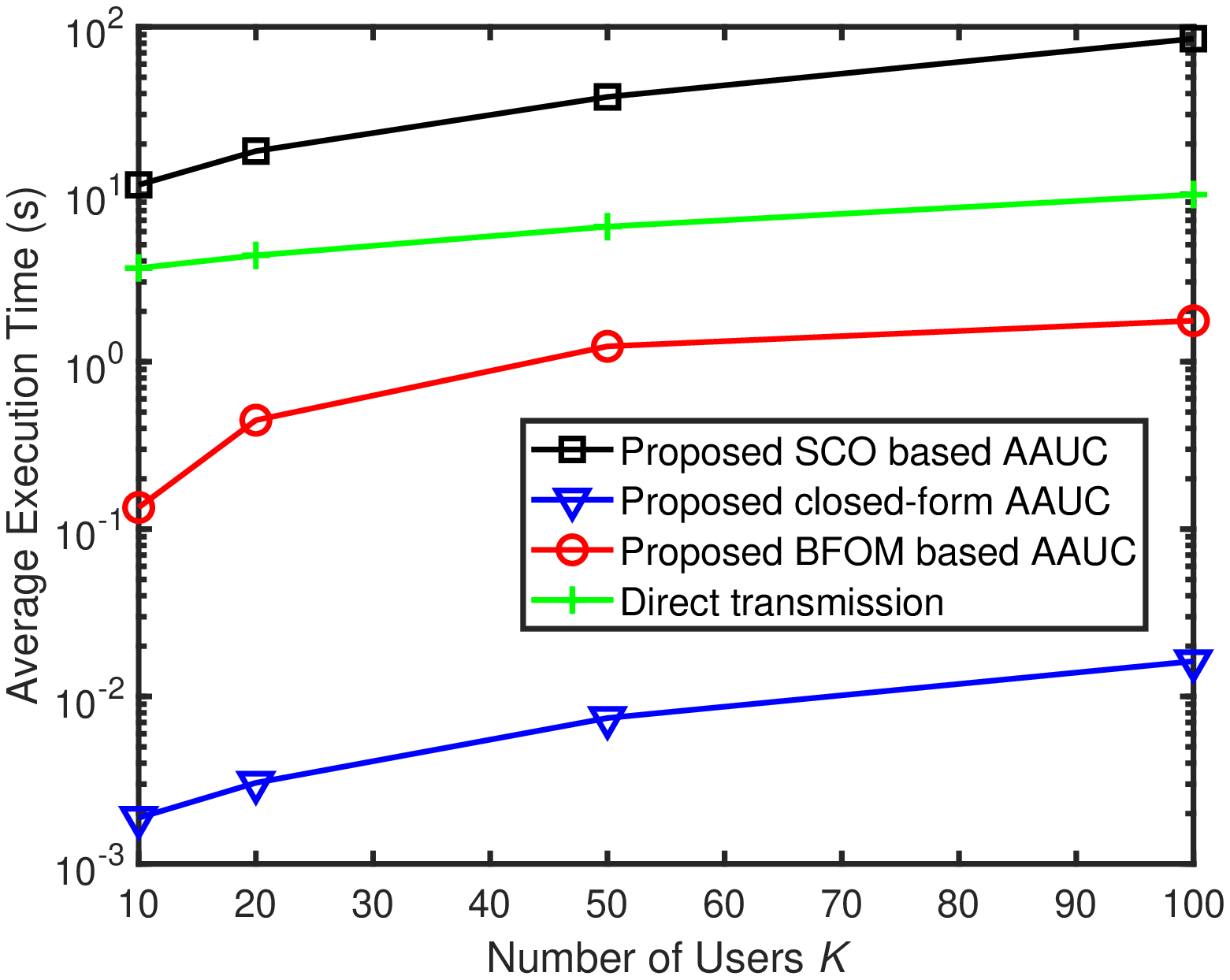}}
\subfigure[]{\includegraphics[width=58mm]{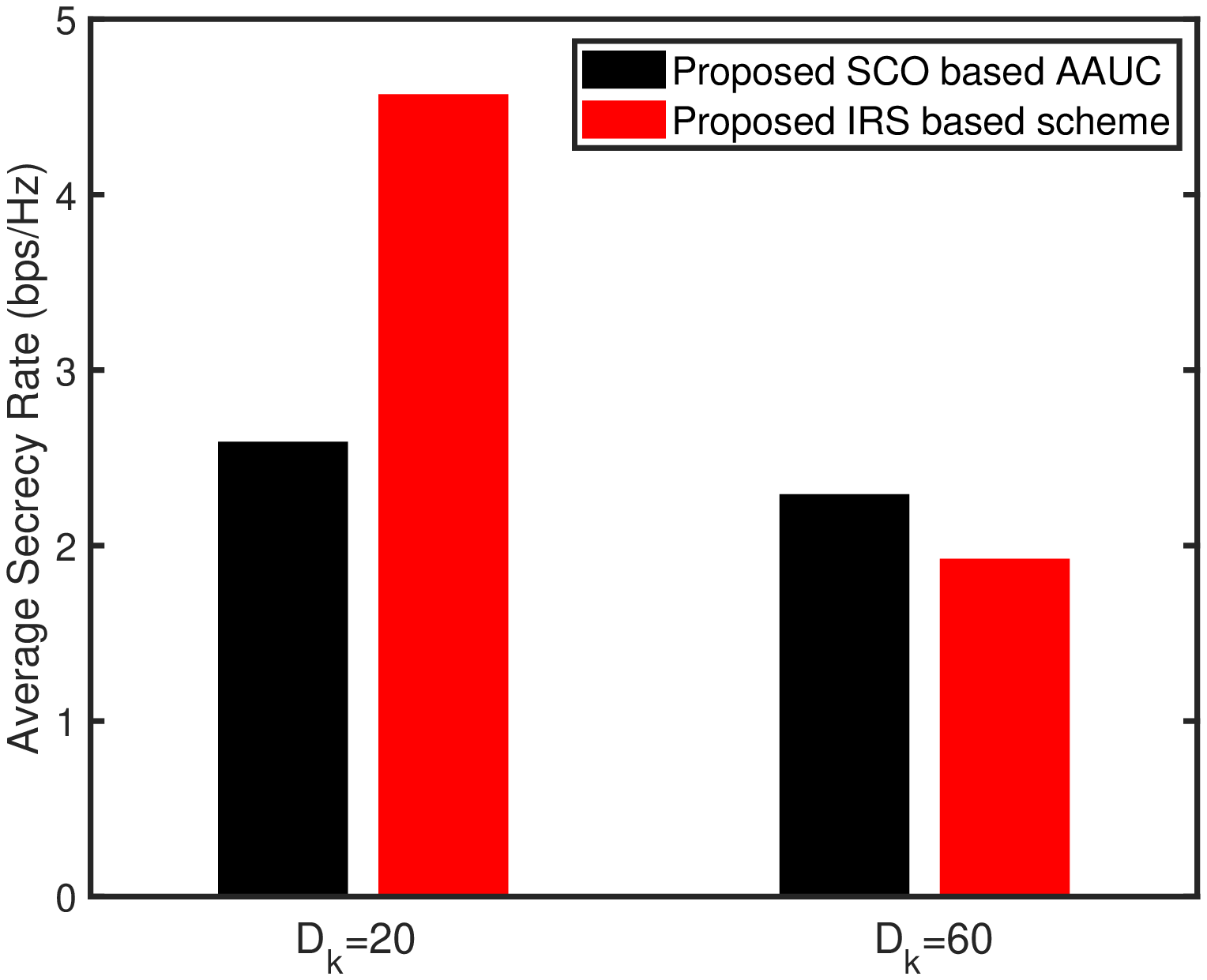}}
\caption{a) Average execution time versus the number of antennas $N$; b) Average execution time versus the number of users $K$; c) Average secrecy rate versus the distance $D_k$ when $N=100$ and $K=10$.}
\end{figure*}

\begin{figure*}[!t]
\centering
\subfigure[]{\includegraphics[width=58mm]{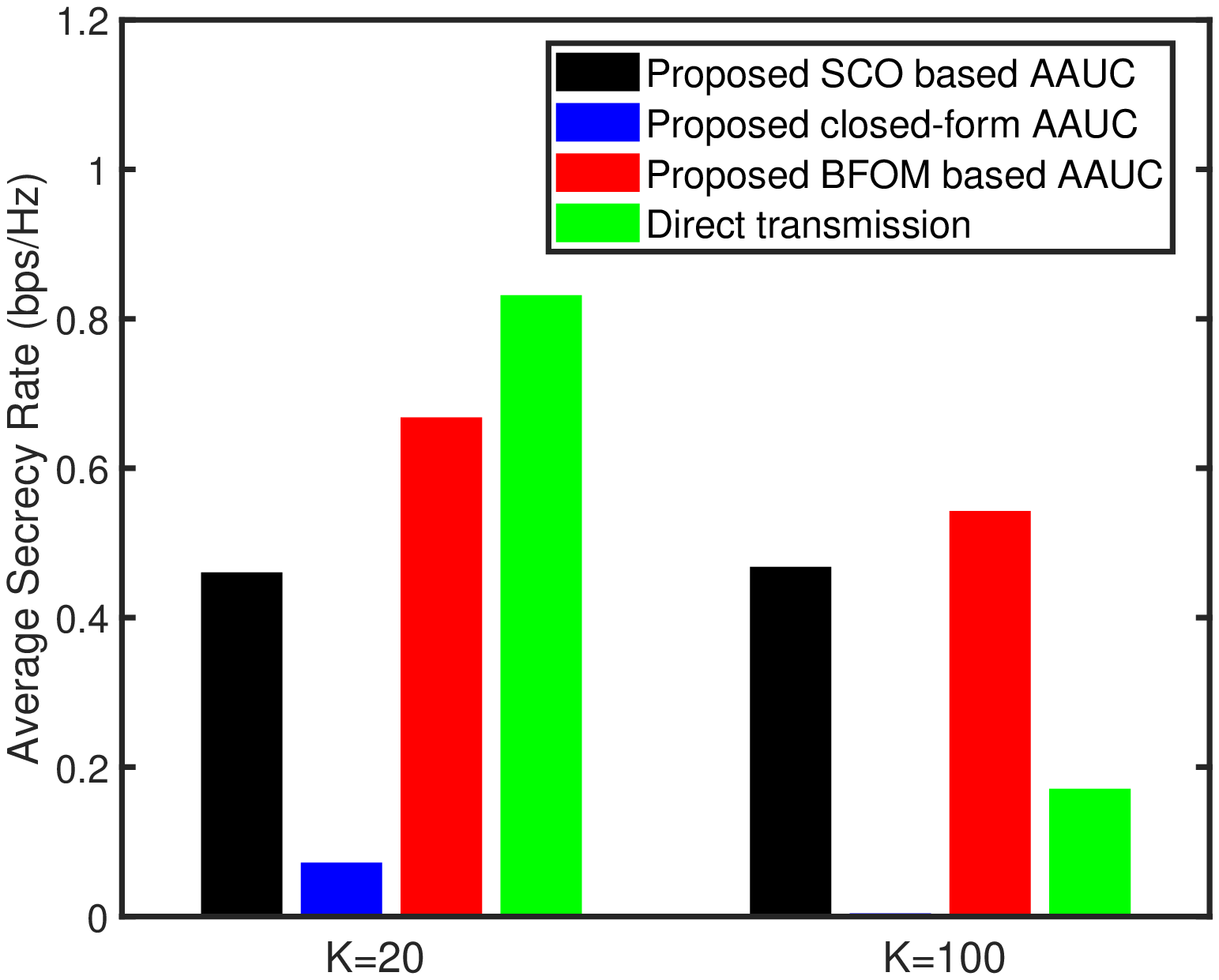}}
\subfigure[]{\includegraphics[width=58mm]{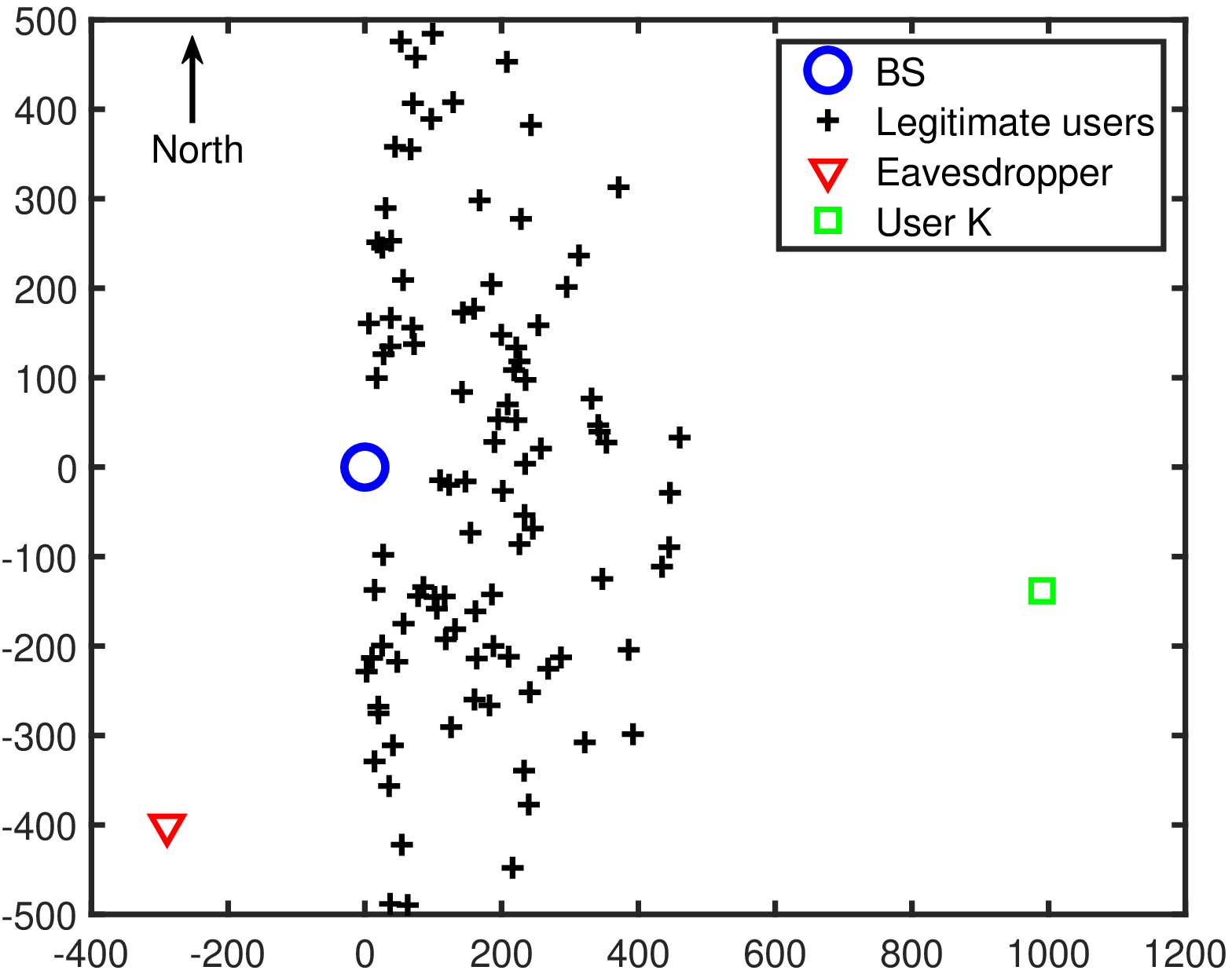}}
\subfigure[]{\includegraphics[width=58mm]{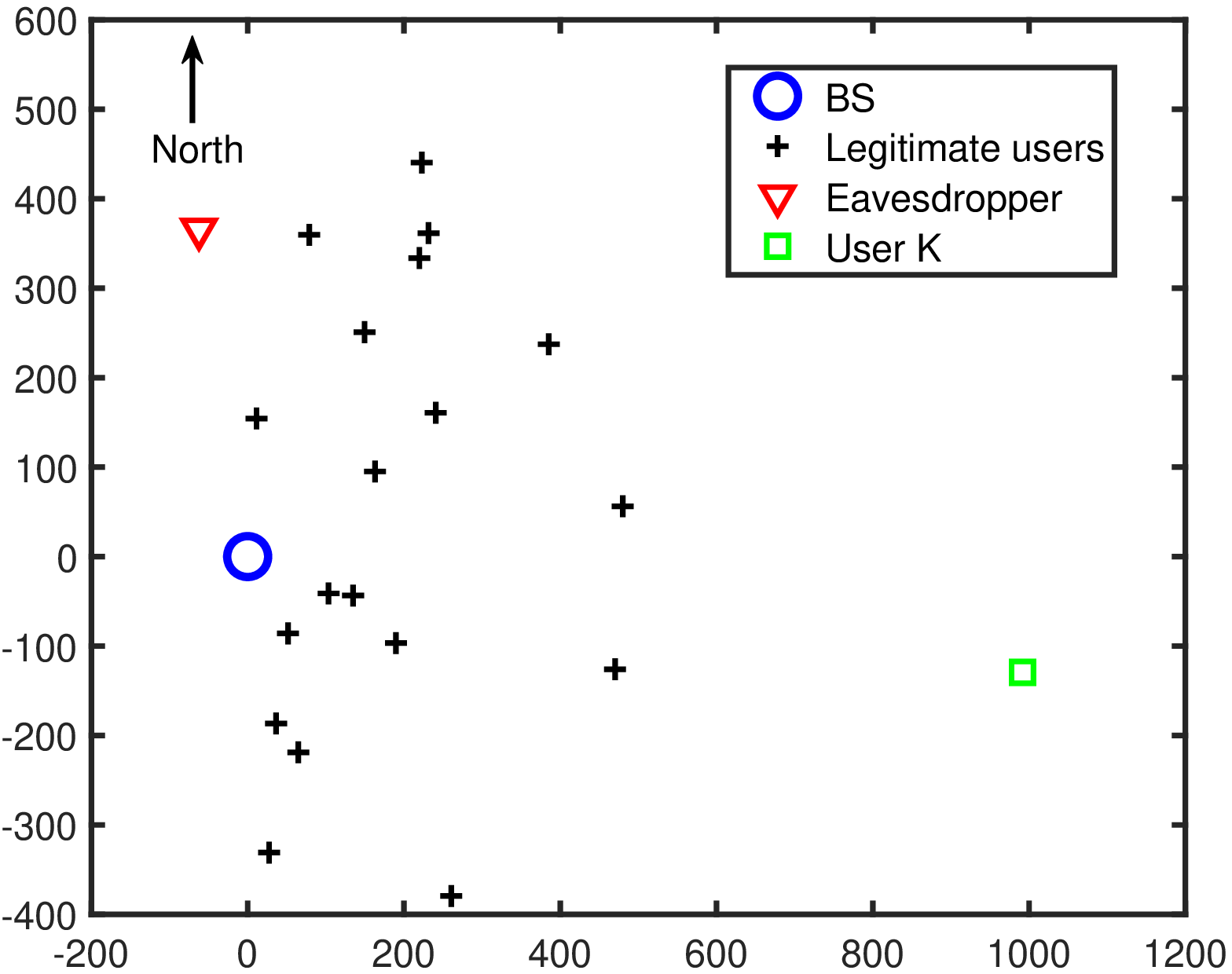}}
\caption{a) Average secrecy rate in Rayleigh fading channels with eavesdropper not in the same angle as one of the users; b) A realization of locations of BS, users, and eavesdropper when $K=100$; c) A realization of locations of BS, users, and eavesdropper when $K=20$.}
\end{figure*}

\subsection{Evaluation of SCO Based AAUC}

In order to verify the performance of the proposed SCO based AAUC, the case of $K=10$ with $N=100$ is simulated, and the average secrecy rate versus Rician K-factor $K_R\in\{0, 10, 20,40\}$ in $\mathrm{dB}$ is shown in Fig.~3a.
Besides the proposed AAUC, we also simulate the direct transmission scheme.
For direct transmission, the multicasting beamforming design $\mathbf{v}$ is obtained by the conic quadratic programming approach \cite{sla}.
It can be seen from Fig. 3a that when $K_R$ is very small, it is still possible to achieve a positive secrecy rate for the direct transmission scheme, since the channels $\{\mathbf{g}_K,\mathbf{g}_E\}$ are not exactly in the same direction due to the non-LOS components.
However, as $K_R$ increases, the direct transmission scheme would lead to zero secrecy rates, since the channels $\mathbf{g}_K$ and $\mathbf{g}_E$ are highly correlated.
This result corroborates our discussions in Section~I that the direct transmission is vulnerable to the eavesdropping in Rician fading environment.
Fortunately, by adopting user cooperation, all the AAUC schemes significantly outperform the direct transmission scheme.
Furthermore, no matter what value $K_R$ takes, the proposed SCO based AAUC always outperforms the closed-form AAUC and the BFOM based AAUC.
This is because SCO based AAUC does not rely on the assumption of large number of antennas or users.
Notice that the secrecy rate of SCO based AAUC increases as $K_R$ increases.
This is because $|\mathbf{g}_K^H\mathbf{v}|=0$ is adopted in $\mathrm{P}1$, and
a higher correlation between $\mathbf{g}_K$ and $\mathbf{g}_E$ means a smaller $|\mathbf{g}_E^H\mathbf{v}|^2$ at the eavesdropper.

\subsection{Evaluation of Large-Scale Optimization Algorithms}

Next, to verify the performance and the low complexity nature of closed-form AAUC when $N$ is large, the case of $K=20$ with $K_R = 30\,\mathrm{dB}$ is simulated, and the average secrecy rate versus the number of antennas $N\in\{20,100,200,500\}$ is shown in Fig. 3b.
It can be seen that when $N$ is small (e.g., $N=20$), the performance gap between SCO based AAUC and closed-form AAUC is large.
This is because the channels $\{\mathbf{g}_k\}$ are not orthogonal for a small $N$, which deviates from the orthogonality assumption adopted in the closed-form AAUC.
However, as $N$ increases, the performance of closed-form AAUC is approaching that of SCO based AAUC.
This corroborates \textbf{Proposition 2} and \textbf{Proposition 3} derived in Section V-A.
On the other hand, Fig.~4a shows the average execution time versus the number of antennas at the BS.
Compared to SCO based AAUC, closed-form AAUC saves the computation times by several orders of magnitude for all the simulated values of $N$, revealing the low-complexity nature of the proposed closed-form solution.

To analyze the performance and the complexity of BFOM based AAUC when $K$ is large, the case of $N=100$ with $K_R=30\,\mathrm{dB}$ is simulated, and the average secrecy rate versus the number of users $K\in\{10, 20, 50,100\}$ is shown in Fig. 3c.
It can be seen that closed-form AAUC performs poorly in this setting since a large $K$ would lead to nonorthogonal channels of $\{\mathbf{g}_k\}$.
In contrast, BFOM based AAUC achieves a satisfactory secrecy rate. 
Moreover, the performance gap between BFOM based AAUC and SCO based AAUC decreases as $K$ increases, which corroborates the discussions in Section V-B.
On the other hand, as shown in Fig.~4b, BFOM based AAUC saves at least $90\%$ (one order of magnitude) of the computation time compared with SCO based AAUC. 
Its flat curve of execution time reveals the linear complexity nature of the proposed BFOM algorithm.

\subsection{Evaluation of IRS Based Scheme}

Now, to evaluate the performance of the proposed IRS based scheme, we simulate the case of $N=100$ and $K=10$ with $K_R=70\,\mathrm{dB}$. 
The average secrecy rate versus the distance $D_k=D_1=D_2=\cdots=D_{10}$ is shown in Fig.~4c.
It can be seen from Fig.~4c that the proposed IRS based scheme outperforms the proposed SCO based AAUC when $D_k=20$ due to the smaller number of transmission phases with IRS.
However, the IRS based scheme performs worse than the proposed SCO based AAUC when $D_k=60$.
This is because the power of the scattered signal from IRS cannot exceed the received signal power at IRS as seen from \eqref{IRS}.
As a result, when the IRS elements are far from the BS, the scattered signal from IRS is weak.
In contrast, the helping users in the SCO based AAUC scheme can use a much larger transmit power than its received signal power.

\subsection{Rayleigh Fading Channel}

Finally, to demonstrate the versatility of the proposed method, we simulate a case of non-Rician fading channel and the eavesdropper is not in the same angle as one of the users.
In particular, we consider the case of $N=100$ and $K_R=0$ (i.e., Rayleigh fading).
The eavesdropper is randomly located at $(D_E\,\mathrm{cos}\,\theta_E,D_E\,\mathrm{sin}\,\theta_E)$ with $D_E\sim \mathcal{U}(100,500)$ and $\theta_E\sim \mathcal{U}(-\pi,0)$.
The users are randomly located at $(D_k\,\mathrm{cos}\,\theta_k,D_k\,\mathrm{sin}\,\theta_k)$ with $\theta_k\sim \mathcal{U}(0,\pi)$, where $D_k\sim \mathcal{U}(100,500)$ for $k\neq K$ and $D_K=1000$.

In this setting, the proposed AAUC scheme still designs $\mathbf{v}$ and $\mathbf{w}$ by solving $\rm{P}1$ (i.e, uses two transmission phases, forces $|\mathbf{g}^H_K\mathbf{v}|=0$, and adopts the angular secrecy model).
The average secrecy rate versus the number of users $K\in\{20,100\}$ is shown in Fig.~5a.
It can be seen that the performance of the closed-form AAUC is not acceptable due to the small $N/K$.
On the other hand, the secrecy rates of SCO based AAUC and BFOM based AAUC are significantly higher than that of direct transmission when $K=100$.
This is because the proposed AAUC exploits path diversity, and the helping users can relay the information to the far-away user as shown in Fig. 5b.
However, the proposed AAUC schemes perform slightly worse than traditional direct transmission when $K=20$ due to $K<N$ (as shown in Fig.~5c).
Fortunately, the proposed method can be executed in combination with direct transmission, and the hybrid strategy can select between executing direct transmission or AAUC.
By switching to the other transmission mode once the current mode is not satisfactory, the hybrid transmission strategy can always achieve high security.

\section{Conclusions}

This paper studied the physical-layer security problem in a massive MIMO multicasting system.
Since the objective function in this problem is not analytical, an angular secrecy model, which matches the numerical integration very well, was proposed.
By adopting this model, secure beamforming design was obtained via the SCO based AUCC algorithm.
In the large-scale settings, two fast algorithms were derived to tackle the curse of high dimensionality.
Simulation results showed that the proposed SCO algorithm achieves higher security than direct transmission.
Furthermore, the proposed fast algorithms significantly reduce the execution time compared to the SCO while still achieving satisfactory performance.

\appendices
\section{Proof of Property (iii)}

To prove this property, we first notice that
\begin{align}
&S(R,\mathbf{w}|\sigma_E)
\nonumber\\
&=\frac{1}{D_K}\int_{0}^{D_K}\mathbb{E}_{\{\phi_k,\delta_k\}}\Bigg\{\Bigg[R
\nonumber\\
&\quad{}
-\frac{1}{2}\mathrm{log}_2
\left(1+\frac{\Big|\mathbf{h}^H_E(\rho,\{\phi_k,\delta_k\})\mathbf{w}
\Big|^2}{\sigma^2_E}\right)\Bigg]^+\Bigg\}\mathrm{d}\rho
\nonumber\\
&\geq
\Bigg[R-\frac{1}{D_K}\int_{0}^{D_K}\mathbb{E}_{\{\phi_k,\delta_k\}}\Bigg\{
\nonumber\\
&\quad{}
\frac{1}{2}\mathrm{log}_2
\left(1+\frac{\Big|\mathbf{h}^H_E(\rho,\{\phi_k,\delta_k\})\mathbf{w}
\Big|^2}{\sigma^2_E}\right)\Bigg\}\mathrm{d}\rho\Bigg]^+
\nonumber\\
&\geq
\Bigg[R-\frac{1}{2}\mathrm{log}_2
\Bigg(1+\frac{1}{\sigma^2_ED_K}
\nonumber\\
&\quad{}
\times
\int_{0}^{D_K}\mathbb{E}_{\{\phi_k,\delta_k\}}\left[\Big|\mathbf{h}^H_E(\rho,\{\phi_k,\delta_k\})\mathbf{w}
\Big|^2\right]\mathrm{d}\rho\Bigg)\Bigg]^+, \label{AppA1}
\end{align}
where the first inequality is due to the convexity of $[x]^+=\mathrm{max}(x,0)$, and the second inequality is due to Jensen's inequality.
On the other hand, we compute
\begin{align}
&
\frac{1}{D_K}\int_{0}^{D_K}\mathbb{E}_{\{\phi_k,\delta_k\}}\left[\Big|\mathbf{h}^H_E(\rho,\{\phi_k,\delta_k\})\mathbf{w}
\Big|^2\right]\mathrm{d}\rho
\nonumber\\
&
=\frac{1}{D_K}\int_{0}^{D_K}\mathbb{E}_{\{\phi_k,\delta_k\}}\Bigg[\Big|\mathop{\sum}_{k=1}^{K-1}
w_k\sqrt{\varrho_0\left(\frac{d_{E,k}(\rho)}{d_0}\right)^{-\alpha}}
\nonumber\\
&\quad{}
\times
\left(\sqrt{\frac{K_R}{1+K_R}}\,\mathrm{e}^{\mathrm{j}\phi_k}
+\sqrt{\frac{1}{1+K_R}}\, \delta_{k}\right)
\Big|^2
\Bigg]\mathrm{d}\rho
\nonumber\\
&=
\frac{1}{D_K}\int_{0}^{D_K}\left[
\mathop{\sum}_{k=1}^{K-1}
|w_k|^2\,\varrho_0\left(\frac{d_{E,k}(\rho)}{d_0}\right)^{-\alpha}
\right]\mathrm{d}\rho
\nonumber\\
&=
\mathop{\sum}_{k=1}^{K-1}
|w_k|^2\,\frac{\varrho_0}{D_K}\int_{0}^{D_K}
\Big[(\rho\,\mathrm{cos}\,\theta_K-D_k\,\mathrm{cos}\,\theta_k)^2
\nonumber\\
&\quad{}
+(\rho\,\mathrm{sin}\,\theta_K-D_k\,\mathrm{sin}\,\theta_k)^2\Big]
^{-\alpha/2}\mathrm{d}\rho
\nonumber\\
&
=\mathbf{w}^H\mathbf{J}\mathbf{w}, \label{AppA2}
\end{align}
where the first equality is obtained by putting $\mathbf{h}_{E}(\rho,\{\phi_k,\delta_k\})$ of \eqref{channel} into $|\mathbf{h}^H_E(\rho,\{\phi_k,\delta_k\})\mathbf{w}|^2$, and the second equality is due to the independence between $\{\phi_k\}$ and $\{\delta_k\}$ together with $\mathbb{E}[|\delta_k|^2]=1$,
and the third equality is obtained from the expression \eqref{dk} of $d_{E,k}(\rho)$.
Combining \eqref{AppA1} and \eqref{AppA2}, the proof is immediately completed.

\section{Proof of Proposition 1}

To prove part (i), consider the following inequality
\begin{align}
\Big(\mathbf{z}-\mathbf{z}^\star\Big)^H\frac{\mathbf{U}^H\mathbf{g}_{k}\mathbf{g}^H_{k}\mathbf{U}}{\sigma_k^2}
\Big(\mathbf{z}-\mathbf{z}^\star\Big)&\geq 0. \label{AppB1}
\end{align}
This always holds due to $\mathbf{U}^H\mathbf{g}_{k}\mathbf{g}^H_{k}\mathbf{U}\succeq \bm{0}$.
Then from \eqref{AppB1} we further have
\begin{align}
\frac{|\mathbf{g}^H_{k}\mathbf{U}\mathbf{z}|^2}{\sigma_k^2}\geq&
2\mathrm{Re}\left[\frac{(\mathbf{z}^\star)^H\mathbf{U}^H\mathbf{g}_{k}\mathbf{g}^H_{k}\mathbf{U}\mathbf{z}}{\sigma_k^2}\right]
-\frac{|\mathbf{g}^H_{k}\mathbf{U}\mathbf{z}^\star|^2}{\sigma_k^2},
\end{align}
which leads to
\begin{align}
&\frac{1}{2}\mathrm{log}_2\left(1+\frac{|\mathbf{g}^H_{k}\mathbf{U}\mathbf{z}|^2}{\sigma_k^2}\right)
\nonumber\\
&
\geq
\frac{1}{2}\mathrm{log}_2\left(1+2\mathrm{Re}\left[\frac{(\mathbf{z}^\star)^H\mathbf{U}^H\mathbf{g}_{k}\mathbf{g}^H_{k}\mathbf{U}\mathbf{z}}{\sigma_k^2}\right]
-\frac{|\mathbf{g}^H_{k}\mathbf{U}\mathbf{z}^\star|^2}{\sigma_k^2}\right). \label{AppB2}
\end{align}
On the other hand, since $-1/2\mathrm{log}_2(\cdot)$ is convex, it must be greater than its first-order Taylor expansion, leading to
\begin{align}
&-\frac{1}{2}\mathrm{log}_2
\left(1+\frac{\lambda\mathbf{w}^H\mathbf{J}\mathbf{w}}{\sigma^2_E}\right)
\geq
-\frac{\lambda\mathbf{w}^H\mathbf{J}\mathbf{w}-\lambda\left(\mathbf{w}^\star\right)^H\mathbf{J}\mathbf{w}^\star}{2\mathrm{ln}2\left[\sigma_E^2+\lambda\left(\mathbf{w}^\star\right)^H\mathbf{J}\mathbf{w}^\star\right]}
\nonumber\\
&~~~~~~~~~~~~~~~~~~~~~~~
-\frac{1}{2}\mathrm{log}_2
\left(1+\frac{\lambda\left(\mathbf{w}^\star\right)^H\mathbf{J}\mathbf{w}^\star}{\sigma^2_E}\right). \label{AppB3}
\end{align}
By adding the equations \eqref{AppB2} and \eqref{AppB3}, the inequality $\widetilde{\Phi}_k(\mathbf{z},\mathbf{w}|\mathbf{z}^\star,\mathbf{w}^\star)\leq \Phi_k(\mathbf{z},\mathbf{w})$ is immediately proved.
Finally, due to $\mathbf{h}_K\mathbf{h}_K^H\succeq \bm{0}$, it is clear that $(\mathbf{w}-\mathbf{w}^\star)^H\frac{\mathbf{h}_K\mathbf{h}_K^H}{\sigma_K^2}
(\mathbf{w}-\mathbf{w}^\star)\geq 0$ and thus
\begin{align}
\frac{|\mathbf{h}_K^H\mathbf{w}|^2}{\sigma_K^2}
\geq&2\mathrm{Re}\left[\frac{(\mathbf{w}^\star)^H\mathbf{h}_K\mathbf{h}_K^H\mathbf{w}}{\sigma_K^2}\right]
-\frac{|\mathbf{h}_K^H\mathbf{w}^\star|^2}{\sigma_K^2}. \label{AppB4}
\end{align}
Combining \eqref{AppB4} and \eqref{AppB3} gives $\widetilde{\Upsilon}(\mathbf{w}|\mathbf{w}^\star)\leq \Upsilon(\mathbf{w})$.

To prove part (ii), we put $\mathbf{z}=\mathbf{z}^\star$ and $\mathbf{w}=\mathbf{w}^\star$ into $\widetilde{\Phi}_k$ in \eqref{Phiapp} and $\widetilde{\Upsilon}$ in \eqref{Upsilonapp}.
Then we immediately obtain $\widetilde{\Phi}_k(\mathbf{z}^\star,\mathbf{w}^\star|\mathbf{z}^\star,\mathbf{w}^\star)=\Phi_k(\mathbf{z}^\star,\mathbf{w}^\star)$ and
$\widetilde{\Upsilon}(\mathbf{w}^\star|\mathbf{w}^\star)=\Upsilon(\mathbf{w}^\star)$.

To prove part (iii), we calculate the following derivatives:
\begin{subequations}
\begin{align}
&\nabla_{\mathbf{z}}\widetilde{\Phi}_k(\mathbf{z},\mathbf{w}|\mathbf{z}^\star,\mathbf{w}^\star)
=\frac{\mathbf{U}^H\mathbf{g}_{k}\mathbf{g}^H_{k}\mathbf{U}\mathbf{z}^\star}{\sigma_k^2\,2\mathrm{ln}2}
\nonumber\\
&
\Bigg(1+2\mathrm{Re}\left[\frac{(\mathbf{z}^\star)^H\mathbf{U}^H\mathbf{g}_{k}\mathbf{g}^H_{k}\mathbf{U}\mathbf{z}}{\sigma_k^2}\right]
-\frac{|\mathbf{g}^H_{k}\mathbf{U}\mathbf{z}^\star|^2}{\sigma_k^2}
\Bigg)^{-1}, \label{AppB5}
\\
&\nabla_{\mathbf{z}}\Phi_k(\mathbf{z},\mathbf{w})=
\left(1+\frac{|\mathbf{g}^H_{k}\mathbf{U}\mathbf{z}|^2}{\sigma_k^2}\right)^{-1}
\frac{\mathbf{U}^H\mathbf{g}_{k}\mathbf{g}^H_{k}\mathbf{U}\mathbf{z}}{\sigma_k^2\,2\mathrm{ln}2},
\\
&\nabla_{\mathbf{w}}\widetilde{\Phi}_k(\mathbf{z},\mathbf{w}|\mathbf{z}^\star,\mathbf{w}^\star)
=
-\frac{1}{2\mathrm{ln}2}\frac{\lambda\mathbf{J}\mathbf{w}}
{\sigma_E^2+\lambda\left(\mathbf{w}^\star\right)^H\mathbf{J}\mathbf{w}^\star},
\\
&\nabla_{\mathbf{w}}\Phi_k(\mathbf{z},\mathbf{w})=
-\frac{1}{2\mathrm{ln}2}\,\frac{\lambda\mathbf{J}\mathbf{w}}
{\sigma_E^2+\lambda\mathbf{w}^H\mathbf{J}\mathbf{w}},
\\
&\nabla_{\mathbf{w}}\widetilde{\Upsilon}(\mathbf{w}|\mathbf{w}^\star)
=\Bigg(1+2\mathrm{Re}\left[\frac{(\mathbf{w}^\star)^H\mathbf{h}_K\mathbf{h}_K^H\mathbf{w}}{\sigma_K^2}\right]
\nonumber\\
&\quad\quad\quad\quad\quad\quad\quad
-\frac{|\mathbf{h}^H_K\mathbf{w}^\star|^2}{\sigma_K^2}\Bigg)^{-1}
\,\frac{\mathbf{h}_K\mathbf{h}_K^H\mathbf{w}^\star}{\sigma_k^2\,2\mathrm{ln}2}
\nonumber\\
&\quad\quad\quad\quad\quad\quad\quad
-\frac{1}{2\mathrm{ln}2}\,\frac{\lambda\mathbf{J}\mathbf{w}}
{\sigma_E^2+\lambda\left(\mathbf{w}^\star\right)^H\mathbf{J}\mathbf{w}^\star},
\\
&\nabla_{\mathbf{w}}\Upsilon(\mathbf{w})=
\left(1+\frac{|\mathbf{h}_K^H\mathbf{w}|^2}{\sigma_K^2}\right)^{-1}\,
\frac{\mathbf{h}_K\mathbf{h}_K^H\mathbf{w}}{\sigma_K^2\,2\mathrm{ln}2}
\nonumber\\
&\quad\quad\quad\quad\quad
-\frac{1}{2\mathrm{ln}2}\,\frac{\lambda\mathbf{J}\mathbf{w}}
{\sigma_E^2+\lambda\mathbf{w}^H\mathbf{J}\mathbf{w}}. \label{AppB6}
\end{align}
\end{subequations}
Then by putting $\mathbf{z}=\mathbf{z}^{\star}$ and $\mathbf{w}=\mathbf{w}^{\star}$ into \eqref{AppB5}--\eqref{AppB6}, the proof for part (iii) is completed.

\section{Proof of Proposition 2}

To prove this proposition, we consider the problem $\mathrm{P}2$ with fixed $\mathbf{w}=\mathbf{w}^*$.
In such a case, the terms $|\mathbf{h}_K^H\mathbf{w}^*|^2$ and $\lambda(\mathbf{w}^*)^H\mathbf{J}\mathbf{w}^*$ in the objective function of $\mathrm{P}2$ are constants, which can be dropped without changing the solution of $\mathbf{z}$.
Further removing  $\frac{1}{2}\mathrm{log}_2(\cdot)$ in the objective function due to its monotonicity, $\mathrm{P}2$ with fixed $\mathbf{w}=\mathbf{w}^*$ is transformed into
\begin{align}
\mathop{\mathrm{max}}_{\substack{\mathbf{z}}}
~&
\mathop{\mathrm{min}}_{k\neq K}~\frac{|\mathbf{g}^H_{k}\mathbf{U}\mathbf{z}|^2}{\sigma_k^2},
\nonumber\\
\mathrm{s. t.}~~&||\mathbf{z}||^2_2\leq 2P_{\rm{max}}-||\mathbf{w}^*||^2_2. \label{AppC1}
\end{align}
To maximize $\mathbf{z}^H\mathbf{U}^H\mathbf{g}_k\mathbf{g}_k^H\mathbf{U}\mathbf{z}$ for $k\neq K$, the optimal $\mathbf{z}^*\in\mathrm{span}\left(\mathbf{U}^H\mathbf{g}_1,\cdots,\mathbf{U}^H\mathbf{g}_{K-1}\right)$.
Therefore, without loss of generality, we can set
\begin{align}
&\mathbf{z}=\sum_{k=1}^{K-1}\sqrt{\xi_k}\,\mathrm{e}^{\mathrm{j}\zeta_k}\frac{\mathbf{U}^H\mathbf{g}_k}{||\mathbf{U}^H\mathbf{g}_k||_2}, \label{v*}
\end{align}
with $\{\xi_k\geq 0\}$ being real nonnegative coefficients and $\{\zeta_k\}$ being the corresponding phases.

When $N\rightarrow\infty$, we have $|\mathbf{g}_j^H\mathbf{g}_k|^2/||\mathbf{g}_k||_2^2\rightarrow 0$ for all $k\neq j$.
Adding to $\mathbf{U}\mathbf{U}^H\rightarrow \mathbf{I}$ as $N\rightarrow\infty$, the quantity $|\mathbf{g}_j^H\mathbf{U}\mathbf{U}^H\mathbf{g}_k|/||\mathbf{U}^H\mathbf{g}_k||_2\rightarrow 0$ for all $k\neq j$.
Therefore, by putting $\mathbf{z}$ in \eqref{v*} into problem \eqref{AppC1}, we obtain $|\mathbf{g}_k^H\mathbf{U}\mathbf{z}|^2=\xi_k||\mathbf{U}^H\mathbf{g}_k||_2^2$ in the objective function and $||\mathbf{z}||_2^2=\sum_{k=1}^{K-1}\xi_k$ in the constraint, meaning that the phases $\{\zeta_k\}$ would not participate in the optimization.
To this end, we can set $\zeta_k=0$ for all $k$ in \eqref{v*}.
Finally, putting \eqref{v*} and $\zeta_k=0$ into \eqref{AppC1}, problem \eqref{AppC1} under $N\rightarrow\infty$ is equivalently written as
\begin{align}
\mathop{\mathrm{max}}_{\substack{\bm{\xi}}}
~&
\mathop{\mathrm{min}}_{k\neq K}~\frac{\xi_k||\mathbf{U}^H\mathbf{g}_{k}||_2^2}{\sigma_k^2},
\nonumber\\
\mathrm{s. t.}~~&\sum_{k=1}^{K-1}\xi_k\leq 2P_{\rm{max}}-||\mathbf{w}^*||^2_2, \label{AppC2}
\end{align}
where $\bm{\xi}=[\xi_1,\cdots,\xi_{K-1}]^T\in\mathbb{R}_+^{(K-1)\times 1}$.
For the above problem, the optimal $\xi_k^*$ must satisfy
\begin{align}
\frac{\xi_1||\mathbf{U}^H\mathbf{g}_{1}||_2^2}{\sigma_1^2}=\cdots=\frac{\xi_{K-1}||\mathbf{U}^H\mathbf{g}_{K-1}||_2^2}{\sigma_{K-1}^2}. \label{AppC3}
\end{align}
Otherwise, we can always increase $\xi_j$ with $j=\mathop{\mathrm{argmin}}_{i\neq K}~\xi_i||\mathbf{U}^H\mathbf{g}_{i}||_2^2/\sigma_i^2$ by a small quantity $\Delta \xi$ and decrease $\xi_j$ with $j=\mathop{\mathrm{argmax}}_{i\neq K}~\xi_i||\mathbf{U}^H\mathbf{g}_{i}||_2^2/\sigma_i^2$ by $\Delta \xi$, which would make $\sum_{k=1}^{K-1}\xi_k$ unchanged but increase the objective function of \eqref{AppC2}.
Combining \eqref{AppC3} and the constraint of \eqref{AppC2}, the optimal $\xi_k^*$ to \eqref{AppC2} is given by
\begin{align}
&\xi_k^*=\frac{ \left(2P_{\rm{max}}-||\mathbf{w}^*||^2_2\right)\sigma_k^2}{\left(\mathop{\sum}_{i=1}^{K-1}\sigma_i^2/||\mathbf{U}^H\mathbf{g}_i||^2_2\right)||\mathbf{U}^H\mathbf{g}_k||_2^2}.
\end{align}
Putting $\xi_k=\xi_k^*$ into \eqref{v*}, the proposition is obtained.

\section{Proof of Proposition 3}

To solve $\mathrm{P}3$, it can be seen that $\Psi(\mathbf{w})$ inside the minimum function is a monotonically decreasing function of $||\mathbf{w}||_2^2$.
Its maximum value $\Psi_{\mathrm{max}}>0$ is obtained at $||\mathbf{w}||_2^2=0$, while its minimum value $\Psi_{\mathrm{min}}<0$  is obtained at $||\mathbf{w}||_2^2=2P_{\rm{max}}$.
On the other hand, $\Upsilon(\mathbf{w})$ inside the minimum function is a monotonically increasing function of $||\mathbf{w}||_2^2$.
Its maximum value $\Upsilon_{\mathrm{max}}>0$ is obtained at $||\mathbf{w}||_2^2=2P_{\rm{max}}$, while its minimum value $\Upsilon_{\mathrm{min}}=0$ is obtained at $||\mathbf{w}||_2^2=0$.

Based on the above observations, we now prove that the optimal $\mathbf{w}^*$ satisfies $\Psi(\mathbf{w}^*)=\Upsilon(\mathbf{w}^*)$ by contradiction.
In particular, assume that $\Psi(\mathbf{w}^*)<\Upsilon(\mathbf{w}^*)$.
Then we must have $\Upsilon(\mathbf{w}^*)>0$ and thus $||\mathbf{w}^*||_2^2>0$.
As a result, we can always decrease $||\mathbf{w}||_2^2$ to increase $\Psi(\mathbf{w})$ such that $\mathrm{min}[\Psi(\mathbf{w}),\Upsilon(\mathbf{w})]$ is increased.
This is in contradiction to $\mathbf{w}^*$ being optimal.
On the other hand, if $\Psi(\mathbf{w}^*)>\Upsilon(\mathbf{w}^*)$, we must have $\Psi(\mathbf{w}^*)>0$ and thus $||\mathbf{w}^*||_2^2<2P_{\rm{max}}$.
Then we can always increase $||\mathbf{w}||_2^2$ to increase $\Upsilon(\mathbf{w})$ such that $\mathrm{min}[\Psi(\mathbf{w}),\Upsilon(\mathbf{w})]$ is increased.
Again, contradiction is established.

Using $\Psi(\mathbf{w}^*)=\Upsilon(\mathbf{w}^*)$, we can add a constraint $\Psi(\mathbf{w})=\Upsilon(\mathbf{w})$ to $\mathrm{P}3$ without changing the problem, and $\mathrm{P}3$ is equivalently written as
\begin{align}
\mathop{\mathrm{max}}_{\substack{\mathbf{w}}}
~\Upsilon(\mathbf{w}),\quad\mathrm{s.t.}~~\Psi(\mathbf{w})=\Upsilon(\mathbf{w}).
\end{align}
By dropping the logarithm function in $\Upsilon$ and rearranging $\Psi(\mathbf{w})=\Upsilon(\mathbf{w})$, the above problem is further simplified into
\begin{align}
\mathop{\mathrm{max}}_{\substack{\mathbf{w}}}
~&\frac{1+|\mathbf{h}_K^H\mathbf{w}|^2/\sigma_K^2}{1+\lambda\mathbf{w}^H\mathbf{J}\mathbf{w}/\sigma^2_E},\quad\mathrm{s.t.}~~
\mathbf{w}^H\mathbf{\Xi}\mathbf{w}=1, \label{AppD1}
\end{align}
where the constraint is obtained from $\Psi(\mathbf{w})=\Upsilon(\mathbf{w})$ and $\mathbf{\Xi}$ is given by \eqref{Xi}.
Using $\mathbf{w}^H\mathbf{\Xi}\mathbf{w}=1$, the objective function of \eqref{AppD1} is equal to
\begin{align}
&\mathbf{w}^H(\mathbf{\Xi}+\mathbf{h}_K\mathbf{h}_K^H/\sigma_K^2)\mathbf{w}
\big/\left[\mathbf{w}^H(\mathbf{\Xi}+\lambda\mathbf{J}/\sigma_E^2)\mathbf{w}\right].
\nonumber
\end{align}
Furthermore, by dropping the constraint $\mathbf{w}^H\mathbf{\Xi}\mathbf{w}=1$, the problem \eqref{AppD1} is relaxed into
\begin{align}
\mathop{\mathrm{max}}_{\substack{\mathbf{w}}}
~&\frac{\mathbf{w}^H(\mathbf{\Xi}+\mathbf{h}_K\mathbf{h}_K^H/\sigma_K^2)\mathbf{w}}{\mathbf{w}^H(\mathbf{\Xi}+\lambda\mathbf{J}/\sigma_E^2)\mathbf{w}}. \label{AppD2}
\end{align}
It can be seen that the objective function of \eqref{AppD2} would not change if we scale $\mathbf{w}$ by any positive factor.
As a result, no matter what value the optimal solution of $\mathbf{w}$ to \eqref{AppD2} takes, we can always scale it such that $\mathbf{w}^H\mathbf{\Xi}\mathbf{w}=1$ is satisfied.
This means that the relaxed problem \eqref{AppD2} is equivalent to \eqref{AppD1}, and we can focus on solving \eqref{AppD2} in the following.

To solve \eqref{AppD2}, a new variable $\mathbf{q}=(\mathbf{\Xi}+\lambda\mathbf{J}/\sigma_E^2)^{1/2}\mathbf{w}$ is introduced to replace $\mathbf{w}$.
Then \eqref{AppD2} is transformed into
\begin{align}
\mathop{\mathrm{max}}_{\substack{\mathbf{q}}}
~&\mathbf{q}^H(\mathbf{\Xi}+\lambda\mathbf{J}/\sigma_E^2)^{-1/2}(\mathbf{\Xi}+\mathbf{h}_K\mathbf{h}_K^H/\sigma_K^2)
\nonumber\\
&\times
(\mathbf{\Xi}+\lambda\mathbf{J}/\sigma_E^2)^{-1/2}\mathbf{q}\Big/\mathbf{q}^H\mathbf{q},
\end{align}
which is the standard eigenvalue problem, and the optimal $\mathbf{q}^*$ is the dominant eigenvector of $(\mathbf{\Xi}+\lambda\mathbf{J}/\sigma_E^2)^{-1/2}(\mathbf{\Xi}+\mathbf{h}_K\mathbf{h}_K^H/\sigma_K^2)
(\mathbf{\Xi}+\lambda\mathbf{J}/\sigma_E^2)^{-1/2}$.
Putting $\mathbf{q}=\mathbf{q}^*$ into $\mathbf{q}=(\mathbf{\Xi}+\lambda\mathbf{J}/\sigma_E^2)^{1/2}\mathbf{w}$ and scaling $\mathbf{w}$ such that $\mathbf{w}^H\mathbf{\Xi}\mathbf{w}=1$, the equation \eqref{w*} is obtained. This completes the proof.

\section{Proof of Proposition 4}

To prove this proposition, we first compute
\begin{align}
&||\nabla_{\mathbf{x}}\Theta^{[n]}
\left(\mathbf{x},\bm{\gamma}\right)-\nabla_{\mathbf{x}}\Theta^{[n]}
\left(\mathbf{x}',\bm{\gamma}\right)||_2
\nonumber\\
&=
\frac{2\gamma_K||\mathbf{h}_K||_2^2}{\sigma_K^2}||\mathbf{x}-\mathbf{x}'||_2
\leq\frac{2||\mathbf{h}_K||_2^2}{\sigma_K^2}||\mathbf{x}-\mathbf{x}'||_2. \label{E1}
\end{align}
It can be seen from \eqref{E1} and \eqref{Lips} that $\widehat{L}$ satisfies \eqref{Lips1}.
Next, we will prove that $\widehat{L}$ satisfies \eqref{Lips2}.
In particular, based on the gradient in \eqref{gradient1}, we have
\begin{align}
&||\nabla_{\mathbf{x}}\Theta^{[n]}
\left(\mathbf{x},\bm{\gamma}\right)-\nabla_{\mathbf{x}}\Theta^{[n]}
\left(\mathbf{x},\bm{\gamma}'\right)||_2
\nonumber\\
&=\Big|\Big|
\mathop{\sum}_{k=1}^{K-1}(\gamma_k-\gamma_k')\mathbf{r}_k^{[n]}
-(\gamma_K-\gamma_K')\frac{2||\mathbf{h}_K||_2^2\mathbf{x}}{\sigma_K^2}
\Big|\Big|_2
\nonumber\\
&
\leq
\mathop{\sum}_{k=1}^{K-1}|\gamma_k-\gamma_k'|\,||\mathbf{r}_k^{[n]}||_2
+|\gamma_K-\gamma_K'|\,\frac{2||\mathbf{h}_K||_2^2||\mathbf{x}||_2}{\sigma_K^2}
\nonumber\\
&
\leq
\mathop{\sum}_{k=1}^{K-1}|\gamma_k-\gamma_k'|\,||\mathbf{r}_k^{[n]}||_2
+|\gamma_K-\gamma_K'|\,\frac{2\sqrt{2P_{\rm{max}}}||\mathbf{h}_K||_2^2}{\sigma_K^2}
\nonumber\\
&
\leq
\widehat{L}
\mathop{\sum}_{k=1}^{K}|\gamma_k-\gamma_k'|
=\widehat{L}||\bm{\gamma}-\bm{\gamma}'||_1, \label{E2}
\end{align}
where the first inequality is due to $||\mathbf{a}+\mathbf{b}||\leq ||\mathbf{a}||+||\mathbf{b}||$,
the second inequality is due to $\mathbf{x}\in\mathcal{A}$,
and the last inequality is due to $||\mathbf{r}_k^{[n]}||_2\leq \widehat{L}$ and $2\sqrt{2P_{\rm{max}}}||\mathbf{h}_K||_2^2/\sigma_K^2\leq \widehat{L}$.

Finally, it is clear that \eqref{Lips3} holds since $||\nabla_{\bm{\gamma}}\Theta^{[n]}
\left(\mathbf{x},\bm{\gamma}\right)-\nabla_{\bm{\gamma}}\Theta^{[n]}
\left(\mathbf{x},\bm{\gamma}'\right)||_{\infty}=0$.
Therefore, we only need to show that $\widehat{L}$ satisfies \eqref{Lips4}.
To this end, based on the gradient in \eqref{gradient2}, the left hand side of \eqref{Lips4} is equal to
\begin{align}
&||\nabla_{\bm{\gamma}}\Theta^{[n]}
\left(\mathbf{x},\bm{\gamma}\right)-\nabla_{\bm{\gamma}}
\Theta^{[n]}
\left(\mathbf{x}',\bm{\gamma}\right)||_{\infty}
\nonumber\\
&
=\mathrm{max}~\Bigg[
\Big|(\mathbf{r}_1^{[n]})^T(\mathbf{x}-\mathbf{x}')\Big|,\cdots,\Big|(\mathbf{r}_{K-1}^{[n]})^T(\mathbf{x}-\mathbf{x}')\Big|,
\nonumber\\
&\quad{}
\frac{||\mathbf{h}_K||_2^2}{\sigma_K^2}\Big|||\mathbf{x}||_2^2-||\mathbf{x}'||_2^2\Big|
\Bigg]
\nonumber\\
&
\leq
\mathrm{max}~\Bigg[
||\mathbf{r}_1^{[n]}||_2\,||\mathbf{x}-\mathbf{x}'||_2,\cdots,||\mathbf{r}_{K-1}^{[n]}||_2\,||\mathbf{x}-\mathbf{x}'||_2,
\nonumber\\
&\quad{}
\frac{||\mathbf{h}_K||_2^2}{\sigma_K^2}\Big|||\mathbf{x}||_2^2-||\mathbf{x}'||_2^2\Big|
\Bigg], \label{E3}
\end{align}
where the inequality is due to $|\mathbf{a}^T\mathbf{b}|\leq ||\mathbf{a}||_2||\mathbf{b}||_2$.
Moreover,
\begin{align}
&\frac{||\mathbf{h}_K||_2^2}{\sigma_K^2}\Big|||\mathbf{x}||_2^2-||\mathbf{x}'||_2^2\Big|
\nonumber\\
&=\frac{||\mathbf{h}_K||_2^2}{\sigma_K^2}\left(||\mathbf{x}||_2+||\mathbf{x}'||_2\right)
\Big|||\mathbf{x}||_2-||\mathbf{x}'||_2\Big|
\nonumber\\
&\leq \frac{||\mathbf{h}_K||_2^2}{\sigma_K^2}\times2\sqrt{2P_{\rm{max}}}
||\mathbf{x}-\mathbf{x}'||_2, \label{E4}
\end{align}
where the inequality is due to $\mathbf{x},\mathbf{x}'\in\mathcal{A}$ and $\Big|||\mathbf{a}||-||\mathbf{b}||\Big|\leq ||\mathbf{a}-\mathbf{b}||$.
By putting \eqref{E4} into \eqref{E3}, the proof is completed.


\begin{thebibliography}{60}

\bibitem{B5G} J.~Zhang, E.~Bj\"{o}rnson, M.~Matthaiou, D.~W.~K.~Ng, H.~Yang, and D.~J. Love, ``Multiple antenna technologies for beyond 5G,'' [Online]. Available: arXiv: 1910.00092, Sep.~2019.

\bibitem{security1} A.~Mukherjee, A.~A. Fakoorian, J.~Huang, and A.~L.~Swindlehurst, ``Principles of physical layer security in multiuser wireless networks: A survey,'' \emph{IEEE Commun. Surveys Tuts.}, vol. 16, no. 3, pp. 1550--1573, 3rd Quarter, 2014.

\bibitem{security2} R.~Lu, H.~Zhu, X.~Liu, J.~K. Liu, and J.~Shao, ``Toward efficient and privacy-preserving computing in big data era,'' \emph{IEEE Network}, vol. 28, no. 4, pp. 46--50, Aug. 2014.

\bibitem{security3} X.~Chen, D.~W.~K.~Ng, W.~Gerstacker, and H.~H.~Chen, ``A survey on multiple-antenna techniques for physical layer security,'' \emph{IEEE Commun. Surveys Tuts.}, vol. 19, no. 2, pp. 1027--1053, Jun. 2017.

\bibitem{secrecyrate1} A.~D.~Wyner, ``The wire-tap channel,'' \emph{Bell Syst. Tech. J.}, vol. 54, no. 8, pp. 1355--1387, Oct. 1975.

\bibitem{secrecyrate2} F.~Oggier and B.~Hassibi, ``The secrecy capacity of the MIMO wiretap channel,'' \emph{IEEE Trans. Inf. Theory}, vol.~57, no.~8, pp.~4961--4972, Aug. 2011.

\bibitem{secrecyrate3} K.~Cumanan, Z.~Ding, M.~Xu, and H.~V.~Poor, ``Secrecy rate optimization for secure multicast communications,'' \emph{IEEE J. Sel. Topics Signal Process.}, vol. 10, no. 8, pp. 1417--1432, Dec. 2016.

\bibitem{secrecyrate4} Z.~Li, S.~Wang, P.~Mu, and Y.-C.~Wu, ``Probabilistic constrained secure transmission: Variable-rate design and performance analysis,'' \emph{IEEE Trans. Wireless Commun.}, early access, Jan.~2020. DOI: 10.1109/TWC.2020.2966476.

\bibitem{massive2} H.~Q. Ngo, H.~Suraweera, M.~Matthaiou, and E.~G.~Larsson, ``Multipair full-duplex relaying with massive arrays and linear processing,'' \emph{IEEE J. Sel. Areas Commun.}, vol.~32, no.~9, pp. 1721--1737, Sep. 2014.

\bibitem{massive1} E.~Bj\"{o}rnson, E.~G.~Larsson, and T.~L.~Marzetta, ``Massive MIMO: Ten myths and one critical question,'' \emph{IEEE Commun. Mag.}, vol. 54, no. 2, pp. 114--123, Feb. 2016.

\bibitem{massive3} J.~Zhang, S.~Chen, Y.~Lin, J.~Zheng, B.~Ai, and L.~Hanzo. ``Cell-free massive MIMO: A new next-generation paradigm,'' \emph{IEEE Access}, vol.~7, pp.~99878--99888, Jul.~2019.

\bibitem{massive_security1} D.~Kapetanovic, G.~Zheng, and F.~Rusek, ``Physical layer security for massive MIMO: An overview on passive eavesdropping and active attacks,''  \emph{IEEE Commun. Mag.}, vol. 53, no. 6, pp. 21--27, Jun. 2015.

\bibitem{massive_security2} Y.~Wu, A.~Khisti, C.~Xiao, G.~Caire, K.-K.~Wong, and X.~Gao, ``A survey of physical layer security techniques for 5G wireless networks and challenges ahead,'' \emph{IEEE J. Sel. Areas Commun.}, vol. 36, no. 4, pp. 679--695, Apr. 2018.

\bibitem{massive_security3} X.~Chen, Z.~Zhang, C.~Zhong, D.~W.~K.~Ng, and R.~Jia, ``Exploiting inter-user interference for secure massive non-orthogonal multiple access,'' \emph{IEEE J. Sel. Areas Commun.}, vol. 36, no. 4, pp. 788--801, Apr. 2018.

\bibitem{active1} B.~Akgun, M.~Krunz, and O.~O.~Koyluoglu, ``Vulnerabilities of massive MIMO systems to pilot contamination attacks,'' \emph{IEEE Trans. Inf. Forensic Secur.}, vol. 14, no. 5, pp. 1251--1263, May 2019.

\bibitem{active2} Y.~O.~Basciftci, C.~E.~Koksal, and A.~Ashikhmin, ``Physical-layer security in TDD massive MIMO,'' \emph{IEEE Trans. Inf. Theory}, vol. 64, no. 11, pp. 7359--7380, Nov. 2018.

\bibitem{active3} Y.~Wu, R.~Schober, D.~W.~K.~Ng, C.~Xiao, and G.~Caire, ``Secure massive MIMO transmission with an active eavesdropper,'' \emph{IEEE Trans. Inf. Theory}, vol. 62, no. 7, pp. 3880--3890, Jul. 2016.

\bibitem{passive1} J.~Zhu, R.~Schober, and V.~K.~Bhargava, ``Linear precoding of data and artificial noise in secure massive MIMO systems,'' \emph{IEEE Trans. Wireless Commun.}, vol. 15, no.~3, pp. 2245--2261, Mar. 2016.

\bibitem{passive2} N.~P.~Nguyen, H.~Q.~Ngo, T.~Q.~Duong, H.~D.~Tuan, and K.~Tourki, ``Secure massive MIMO with the artificial noise-aided downlink training,'' \emph{IEEE J. Sel. Areas Commun.}, vol. 36, no. 4, pp. 802--816, Apr. 2018.

\bibitem{passive3} Y.-Y.~Zhang, J.-K.~Zhang, and H.-Y.~Yu, ``Physically securing energy-based massive MIMO MAC via joint alignment of multi-user constellations and artificial noise,'' \emph{IEEE J. Sel. Areas Commun.}, vol. 36, no. 4, pp. 829--844, Apr. 2018.

\bibitem{passive4} F.~Zhu, F.~Gao, H.~Lin, S.~Jin, J.~Zhao, and G.~Qian, ``Robust beamforming for physical layer security in BDMA massive MIMO,'' \emph{IEEE J. Sel. Areas Commun.}, vol. 36, no. 4, pp. 775--787, Apr. 2018.

\bibitem{jzhang}  J.~Zhang, L.~Dai, X.~Zhang, E.~Bj\"{o}rnson, and Z.~Wang, ``Achievable rate of Rician large-scale MIMO channels with transceiver hardware impairments,'' \emph{IEEE Trans. Veh. Technol.}, vol.~65, no.~10, pp. 8800--8806, Oct.~2016.

\bibitem{jzhang2} ￼J.~Zhang, L.~Dai, X.~Li, Y.~Liu, and L.~Hanzo, ``On low-resolution ADCs in practical 5G millimeter-wave massive MIMO systems,'' \emph{IEEE Commun. Mag.}, vol.~56, no.~7, pp.~205--211, Jul.~2018.

\bibitem{hyang} H.~Yang and T.~L.~Marzetta, ``Massive MIMO with max-min power control in line-of-sight propagation environment,'' \emph{IEEE Trans. Commun.}, vol.~65, no.~11, pp. 4685--4693, Nov.~2017.

\bibitem{detection} A.~Chaman, J.~Wang, J.~Sun, H.~Hassanieh, and R.~R.~Choudhury, ``Ghostbuster: Detecting the presence of hidden eavesdroppers,'' in {\it Proc. ACM Mobicom}, New Delhi, India, 2018, pp. 337--351.

\bibitem{channel1} A.~A.~Khuwaja, Y.~Chen, N.~Zhao, M.~S.~Alouini, and P.~Dobbins, ``A survey of channel modeling for UAV communications,'' \emph{IEEE Commun. Surveys Tuts.}, vol. 20, no. 4, pp. 2804--2821, Jul. 2018.

\bibitem{channel2} J.~Zhang, L.~Dai, Z.~He, S.~Jin, and X.~Li, ``Performance analysis of mixed-ADC massive MIMO systems over Rician fading channels,'' \emph{IEEE J. Sel. Areas Commun.}, vol. 35, no. 6, pp. 1327--1338, Jun.~2017.

\bibitem{channel3} M.~Malmirchegini and Y.~Mostofi, ``On the spatial predictability of communication channels,'' \emph{IEEE Trans. Wireless Commun.}, vol. 11, no. 3, pp. 964--978, Mar. 2012.

\bibitem{los} E.~Karipidis, N.~D.~Sidiropoulos, and Z.~Q.~Luo, ``Far-field multicast beamforming for uniform linear antenna arrays,'' \emph{IEEE Trans. Signal Process.}, vol 55, no. 10, pp. 4916--4927, Oct. 2007.

\bibitem{relay} J.~N.~Laneman, D.~N.~C. Tse, and G.~W. Wornell, ``Cooperative diversity in wireless networks: Efficient protocols and outage behavior,'' \emph{IEEE Trans. Inf. Theory}, vol.~50, no.~12, pp.~3062--3080, Dec.~2004.

\bibitem{robust1} D.~W.~K.~Ng, E.~S.~Lo, and R.~Schober, ``Robust beamforming for secure communication in systems with wireless information and power transfer,'' \emph{IEEE Trans. Wireless Commun.}, vol. 13, no. 8, pp. 4599--4615, Aug. 2014.

\bibitem{averagesecrecy1} A.~Rabbachin, A.~Conti, and M.~Z.~Win, ``Wireless network intrinsic secrecy,'' \emph{IEEE/ACM Trans. Netw.}, vol.~23, no.~1, pp.~56--69, Feb.~2015.

\bibitem{averagesecrecy2}  A.~Mukherjee, ``Physical-layer security in the Internet of Things: Sensing and communication confidentiality under resource constraints,'' \emph{Proc. IEEE}, vol.~103, no.~10, pp.~1747--1761, Oct.~2015.

\bibitem{averagesecrecy3} L.~Mucchi, L.~Ronga, X.~Zhou, K.~Huang, Y.~Chen, and R.~Wang, ``A new metric for measuring the security of an environment: The secrecy pressure,'' \emph{IEEE Trans. Wireless Commun.}, vol.~16, no.~5, pp.~3416--3430, May~2017.

\bibitem{3d} X.~Sun, D.~W.~K.~Ng, Z.~Ding, Y.~Xu, and Z.~Zhong, ``Physical layer security in UAV systems: Challenges and opportunities,'' \emph{IEEE Wireless Commun.}, vol.~26, no.~5, pp.~40--47 Oct.~2019.

\bibitem{sco1}B.~R.~Marks and G.~P.~Wright, ``A general inner approximation algorithm for nonconvex mathematical programs,'' \emph{Operation Research}, vol.~26, no.~4, Jul.~1978.

\bibitem{sco2} T.~Lipp and S.~Boyd, ``Variations and extensions of the convex-concave procedure,'' \emph{Optimization and Engineering}, vol.~17, no.~2, pp.~263--287, Jun.~2016.

\bibitem{sco3} M.~F.~Hanif, Z.~Ding, T.~Ratnarajah, and G.~K.~Karagiannidis, ``A minorization-maximization method for optimizing sum rate in non-orthogonal multiple access systems,'' \emph{IEEE Trans. Signal Process.}, vol.~64, no.~1, pp.~76--88, Jan.~2016.

\bibitem{sco4} Y.~Sun, P.~Babu, and D.~P.~Palomar, ``Majorization-minimization algorithms in signal processing, communications, and machine learning,'' \emph{IEEE Trans. Signal Process.}, vol.~65, no.~3, pp.~794--816, Feb.~2017.

\bibitem{sco5} H.~A.~Le~Thi and P.~D.~Tao, ``The DC (difference of convex functions) programming and DCA revisited with DC models of real world nonconvex optimization problems,'' \emph{Annals of Operations Research}, vol.~133, no.~1--4, pp. 23--46, Jan.~2005.

\bibitem{sco6} S.~Wang, M.~Xia, and Y.-C.~Wu, ``Multipair two-way relay network with harvest-then-transmit users: resolving pairwise uplink-downlink coupling,'' \emph{IEEE J. Sel. Topics Signal Process.}, vol.~10, no.~8, pp.~1506--1521, Dec.~2016.

\bibitem{opt2} S.~Boyd and L.~Vandenberghe, \emph{Convex Optimization}. Cambridge, U.K.: Cambridge Univ. Press, 2004.

\bibitem{opt1} A.~Ben-Tal and A.~Nemirovski, \emph{Lectures on Modern Convex Optimization} (MPS/SIAM Series on Optimizations). Philadelphia, PA, USA: SIAM, 2013.

\bibitem{BFOM1} A.~Nemirovski, ``Prox-method with rate of convergence $\mathcal{O}(1/T)$ for variational inequalities with Lipschitz continuous monotone operators and smooth convex-concave saddle point problems,'' \emph{SIAM J. Optimiz.}, vol. 15, no. 1, pp. 229--251, 2004.

\bibitem{BFOM2} S.~Bubeck, ``Convex optimization: Algorithms and complexity,'' \emph{Foundations and Trends in Machine Learning}, vol. 8, no. 3--4, pp. 231--357, Nov. 2015.

\bibitem{BFOM3} A.~Konar and N.~D.~Sidiropoulos, ``Fast approximation algorithms for a class of non-convex QCQP problems using first-order methods,''
\emph{IEEE Trans. Signal Process.}, vol. 65, no. 13, pp. 3494--3509, Jul. 2017.

\bibitem{BFOM4} S.~Wang, Y.-C.~Wu, M.~Xia, R.~Wang, and H.~V.~Poor, ``Machine intelligence at the edge with learning centric power allocation,'' Nov.~2019. [Online]. Available: https://arxiv.org/pdf/1911.04922.pdf.

\bibitem{irs1} Q.~Wu and R.~Zhang, ``Towards smart and reconfigurable environment: Intelligent reflecting surface aided wireless networks,'' \emph{IEEE Commun. Mag.}, vol.~58, no.~1, pp.~106--112, Jan.~2019.

\bibitem{irs2} Q.~Wu and R.~Zhang, ``Intelligent reflecting surface enhanced wireless network via joint active and passive beamforming,'' \emph{IEEE Trans. Wireless Commun.}, vol.~18, no.~11, pp.~5394--5409, Aug.~2019.

\bibitem{irs3} X.~Yu, D.~Xu, Y,~Sun, D.~W.~K.~Ng, and R.~Schober, ``Robust and secure wireless communications via intelligent reflecting surfaces,'' [Online]. Available: arXiv:1912.01497, Dec.~2019. 

\bibitem{irs4} D.~Xu, X.~Yu, Y.~Sun, D.~W.~K.~Ng, and R.~Schober, ``Resource allocation for secure IRS-assisted multiuser MISO systems,'' \emph{IEEE Globecom Workshops}, Dec.~2019. [Online]. Available: arXiv:1907.03085.

\bibitem{wang} S.~Wang, M.~Xia, K.~Huang, and Y.-C.~Wu, ``Wirelessly powered two-way communication with nonlinear energy harvesting model: Rate regions under fixed and mobile relay,'' \emph{IEEE Trans. Wireless Commun.}, vol.~16, no.~12, pp.~8190--8204, Dec.~2017.

\bibitem{3gpp} 3GPP, ``Technical specification group radio access network; study on 3d channel model for LTE (release 12),'' TR 36.873 V12.7.0, Dec. 2017.

\bibitem{sla} L.~N. Tran, M.~F. Hanif, and M.~Juntti, ``A conic quadratic programming approach to physical layer multicasting for large-scale antenna arrays,'' \emph{IEEE Signal Process. Lett.}, vol. 21, no. 1, pp. 114--117, Jan. 2014.

\end{thebibliography}
\end{document}